\documentclass[a4paper,10pt]{article}
\usepackage{amssymb,amsmath,amsthm,amsfonts,amsxtra,enumerate, mathtools,mathrsfs}
\usepackage{fullpage}
\usepackage{hyperref}
\usepackage{eucal}
\usepackage{color}
\usepackage{xcolor}
\usepackage{stackrel}
\usepackage{graphicx}
\newcommand{\Cb}{\mathbb{C}}
\newcommand{\Rb}{\mathbb{R}}

\newcommand{\Wb}{\mathbb{W}}

\newcommand{\ii}{\mathrm{i}}
\newcommand{\ee}{\mathrm{e}}
\newcommand{\dd}{\,\mathrm{d}}
\renewcommand{\phi}{\varphi}
\renewcommand{\theta}{\vartheta}
\newcommand{\p}{\partial}
\newcommand{\ie}{\emph{i.e.}}

\newcommand{\cf}{\emph{cf.}}
\newcommand{\etal}{\emph{et al.}}

\newcommand{\Real}{\mathbb{R}}
\newcommand{\Com}{\mathbb{C}}
\newcommand{\Nat}{\mathbb{N}}
\newcommand{\Int}{\mathbb{Z}}

\newcommand{\supp}{\mathop{\mathrm{supp}}\nolimits}
\newcommand{\Dom}{\mathsf{D}}

\newcommand{\diag}{\mathop{\mathrm{diag}}\nolimits}
\newcommand{\curl}{\mathop{\mathrm{curl}}\nolimits}
\newcommand{\ph}{\mathop{\mathrm{ph}}\nolimits}
\newcommand{\eps}{\varepsilon}
\newcommand{\sii}{L^2}

\newtheorem{Theorem}{Theorem}
\newtheorem{Lemma}{Lemma}
\newtheorem{Proposition}{Proposition}
\newtheorem{Corollary}{Corollary}

\newtheorem{Assumption}{Assumption}
\theoremstyle{definition}
\newtheorem{Remark}{Remark}

\usepackage[normalem]{ulem}
\definecolor{DarkGreen}{rgb}{0,0.5,0.1} 

\newcommand\soutD{\bgroup\markoverwith
{\textcolor{DarkGreen}{\rule[.5ex]{2pt}{1pt}}}\ULon}
\newcommand\soutP{\bgroup\markoverwith
{\textcolor{blue}{\rule[.5ex]{2pt}{1pt}}}\ULon}
\newcommand{\Hm}[1]{\leavevmode{\marginpar{\tiny%
$\hbox to 0mm{\hspace*{-0.5mm}$\leftarrow$\hss}%
\vcenter{\vrule depth 0.1mm height 0.1mm width \the\marginparwidth}%
\hbox to
0mm{\hss$\rightarrow$\hspace*{-0.5mm}}$\\\relax\raggedright #1}}}

\begin{document}
%
%
\title{\textbf{\LARGE
Virtual bound states of the Pauli operator 
with an Aharonov--Bohm potential
}}
\author{Marie Fialov\'a\,$^a$ \ and \ David Krej\v{c}i\v{r}{\'\i}k\,$^b$}
\date{\small 
\vspace{-5ex}
\begin{quote}
\emph{
\begin{itemize}
\item[$a)$] 
Institute of Science and Technology Austria,
Am Campus 1, 3400 Klosterneuburg, Austria;
Marie.Fialova@ist.ac.at
\item[$b)$] 
Department of Mathematics, Faculty of Nuclear Sciences and 
Physical Engineering, Czech Technical University in Prague, 
Trojanova 13, 12000 Prague 2, Czechia;
David.Krejcirik@fjfi.cvut.cz.%
\end{itemize}
}
\end{quote}
28 January 2025
}
\maketitle
\begin{abstract}
A maximal realisation of the two-dimensional Pauli operator,
subject to Aharonov--Bohm magnetic field, is investigated. 
Contrary to the case of the Pauli operator 
with regular magnetic potentials, 
it is shown that both components of the Pauli operator are critical.
Asymptotics of the weakly coupled eigenvalues, 
generated by electric (not necessarily self-adjoint) perturbations,
are derived.
\end{abstract}
%

\section{Introduction}
%
In non-relativistic quantum mechanics, 
a free spinless particle in the plane is described 
by the self-adjoint realisation of the Laplacian $-\Delta$ in $\sii(\Real^2)$.
It is well known that the two-dimensional Laplacian is \emph{critical},
in the sense that its spectrum is unstable under small perturbations.
More specifically, $-\Delta + v$ possesses a negative eigenvalue 
whenever $v \in C_0^\infty(\Real^2)$ is attractive 
(\ie, non-trivial and non-positive).  
In physical terms, $-\Delta$ admits a \emph{virtual bound state} at zero energy,
meaning that, while the spectrum of $-\Delta$ is purely absolutely continuous,
the singularity of the Green function in the spectral parameter 
leads to the genuine eigenvalue under 
the arbitrarily small attractive electric perturbations~$v$.
More specifically, the \emph{weakly coupled} asymptotics 
\begin{equation}\label{exponential}
  \inf\sigma(-\Delta + \eps v) 
  \sim - 
  \mbox{$
  \exp\left(\left[
  \frac{\eps}{4\pi} \int_{\Real^2} v\right]^{-1}
  \right) 
  $}
  \qquad \mbox{as} \qquad 
  \eps \to 0^+
\end{equation}
holds true.
As the amount of literature on the subject is vast,
we mention only~\cite{Si76} for the pioneering work,
\cite{BGS77,KS80} for the earlier papers
and \cite{HS10,CM21,CD24,CDG24} for the most recent contributions.
  
Switching on the magnetic field, the situation changes dramatically
because of the \emph{diamagnetic} effect. 
Mathematically, the magnetic Laplacian $-\Delta_a := (-i\nabla-a)^2$ 
in $\sii(\Real^2)$ becomes \emph{subcritical} 
for any smooth vector potential $a :\Real^2\to\Real^2$
such that the magnetic field $b := \curl a$ is non-trivial. 
Here the subcriticality means the existence of a Hardy inequality
$-\Delta_a \geq \rho$ with positive $\rho:\Real^2 \to \Real$,
which is equivalent to the fact that 
the singularity of the Green function disappears. 
Then $\inf\sigma(-\Delta_a + \eps v) \geq 0$ 
for all sufficiently small~$\eps$,
so there are no negative weakly coupled eigenvalues. 
In fact, this change of game equally applies to the singular 
Aharonov--Bohm potential 
\begin{equation}\label{potential}
  a_\alpha(x) := \alpha \, \frac{(x^2,-x^1)}{|x|^2}
  \qquad \mbox{with} \qquad 
  \alpha \in \Real
\end{equation}
whenever $\alpha \not\in \Int$.
This is particularly spectacular 
because $b_\alpha := \curl a_\alpha$ vanishes almost everywhere in the plane.
In fact, $b_\alpha = -2\pi\alpha \delta$ 
in the sense of distributions, where~$\delta$ is the zero-centred Dirac function.
Again, we mention only~\cite{Laptev-Weidl_1999} for the pioneering work
and \cite{Pankrashkin-Richard_2011,
CK,CFK,Fanelli-Zhang-Zheng_2023,FK24,Fer22,CF23,Fermi_2024} 
for the most recent contributions.

For particles with spin, however, a more realistic (yet still non-relativistic)
description is through the Pauli operator 
\begin{equation}\label{Hamiltonian}
  H_a :=
  \begin{pmatrix}
    -\Delta_a-b & 0 \\
    0 & -\Delta_a + b
  \end{pmatrix} 
  \qquad \mbox{in} \qquad 
  \sii(\Real^2,\Com^2)
  \,,
\end{equation}
subject to matrix-valued potential perturbations 
$V:\Real^2 \to \Com^{2 \times 2}$.
It turns out that these systems exhibit the \emph{paramagnetic} effect.
Indeed, it is known~\cite{Wei99}
(though perhaps less than in the magnetic-free spinless case above)
that~$H_a$ with any regular potential~$a$
does admit a virtual bound state at zero energy. 
Contrary to the superfast exponential decay~\eqref{exponential}, however,
the weak coupling is stronger now:
\begin{equation}\label{polynomial}
  \inf\sigma(H_a + \eps v I_{\Com^2}) 
  \sim - C \, \eps^{1/|\Phi_b|} 
  \qquad \mbox{as} \qquad 
  \eps \to 0^+
  \,,
\end{equation}
where~$C$ is a constant depending on~$v$ and~$b$,
$
  \Phi_b := \frac{1}{2\pi} \int_{\Real^2} b
$
is the total flux of~$b$ and it is assumed that $|\Phi_b| \in (0,1)$.
The eigenvalue  asymptotics~\eqref{polynomial} 
are due to~\cite{Frank-Morozov-Vugalter_2011}
and~\cite{Kovarik_2022,Baur} in radial and general cases, respectively.
(See also \cite{Frank-Kovarik} for related Lieb--Thirring inequalities.)
Of course, the existence of the one virtual bound state 
(and thus correspondingly unique weakly coupled eigenvalue)
is due to the fact that one (and only one) of the operators 
$-\Delta_a \mp b$ appearing in~\eqref{Hamiltonian} is critical
(the choice~$\mp$ depends on the sign of~$\Phi_b$,
namely $-\Delta_a - b$ is critical if $\Phi_b > 0$).

The objective of this paper is to analyse the criticality properties
and weakly coupled eigenvalues of the Pauli operator~\eqref{Hamiltonian}
in the highly singular situation of 
the Aharonov--Bohm potential~\eqref{potential},
formally acting by 
\begin{equation}\label{AB-Hamiltonian}
  H_{a_\alpha} :=
  \begin{pmatrix}
    -\Delta_{a_\alpha}+2\pi\alpha \delta & 0 \\
    0 & -\Delta_{a_\alpha} - 2\pi\alpha \delta
  \end{pmatrix} 
  \qquad \mbox{in} \qquad 
  \sii(\Real^2,\Com^2)
  \,.
\end{equation}
Note that $\Phi_{b_\alpha} = -\alpha$ for the choice~\eqref{potential}.
The operator~$H_{a_\alpha}$ is rigorously introduced by considering 
a self-adjoint realisation of the Aharonov–Bohm Laplacian 
with delta-type interactions 
$H_\alpha^\pm := -\Delta_{a_\alpha} \pm 2\pi\alpha \delta$
due to \v{S}\v{t}ov\'i\v{c}ek \etal\
\cite{DS98,Geyler-Stovicek_2004,Geyler-Stovicek_2004_two_fluxes}. 

Our motivation is two-fold. 
First, contrary to the case of regular potentials 
considered in~\cite{Frank-Morozov-Vugalter_2011,Kovarik_2022,Baur},
the resolvent kernel of the unperturbed Hamiltonian~$H_{a_\alpha}$
is known explicitly. 
Consequently, the standard Birman--Schwinger analysis 
for the study of weakly coupled eigenvalues is available,
so it is possible to avoid the radial hypothesis on~$v$ 
as well as the necessity of advanced resolvent expansions 
due to~\cite{Frank-Morozov-Vugalter_2011} 
and~\cite{Kovarik_2022,Baur}, respectively.
Second, inspired by a recent interest in quantum Hamiltonians
with complex electromagnetic fields,
our setting allows for considering complex-valued~$v$
(and in fact more general matrix-valued perturbations~$V$,
though we do not pursue this research in this paper). 
In summary, comparing with the precedent 
papers~\cite{Frank-Morozov-Vugalter_2011,Kovarik_2022,Baur},
our choice of the magnetic potential is special,
but the electric perturbations are allowed to be more general.

In agreement with the case of regular magnetic potentials,
the main result of our analysis shows~$H_{a_\alpha}$ is critical
whenever $\alpha \not\in \Int$ (by gauge invariance,
one may restrict to $\alpha \in (0,1)$).  
However, there are always \emph{two} virtual bound states now.
Indeed, a variant of our main result can be stated as follows.
\begin{Theorem}\label{Thm.main}
Let $\alpha \in (0,1)$
and $v \in C_0^\infty(\Real^2)$ be non-trivial and non-positive.
Then
\begin{equation}\label{ours}
\begin{aligned}
  \inf\sigma(H_\alpha^+ + \eps v) 
  &\sim - C^+ \, \eps^{1/(1-\alpha)} 
  \,,
  \\
  \inf\sigma(H_\alpha^- + \eps v) 
  &\sim - C^- \, \eps^{1/\alpha} 
  \,,
\end{aligned}  
  \qquad \mbox{as} \qquad 
  \eps \to 0^+
  \,,
\end{equation}
where~$C^\pm$ are positive constants depending on~$v$ and~$\alpha$.
\end{Theorem}

The asymptotics~\eqref{ours} 
remain essentially the same for complex-valued~$v$,
under suitable hypotheses imposed on integrals involving~$v$ 
(the extension to non-diagonal matrix-valued perturbations~$V$
of~$H_{a_\alpha}$ is left open in this paper). 
At the same time, we substantially relax 
the regularity and sign restrictions on real-valued potentials~$v$ below.
See Theorem~\ref{thm:asymptotics} for our main general result.

The existence of two weakly coupled eigenvalues~\eqref{ours},
contrary to the merely one in the case of regular potentials considered 
in~\cite{Frank-Morozov-Vugalter_2011,Kovarik_2022,Baur},
might seem controversial at a first glance,
arguing that the Aharonov--Bohm potential can be approximated 
by the regular potentials.
However, this approximation on the operator level is known  
to be delicate~\cite{Mor96,BV93,BP03,Per05,Per06} 
(see also~\cite{Tam03} for Dirac operators).  
As a matter of fact, there exist several realisations 
of the Pauli operator~\eqref{AB-Hamiltonian}
with the Aharonov--Bohm potential~\eqref{AB-Hamiltonian}
(see~\cite{BCF23} for a recent complete study),
each has its pro and con from the point of view of physical properties,
however, no common agreement on the right choice 
seem to exist in the community.
Our paper provides a negative answer to 
a question of Persson's \cite[Sec.~4]{Per05}
about the possibility of the approximation by regular potentials
for our self-adjoint realisation.
On the positive side, comparing the constant~$C^-$
from the first of our asymptotics~\eqref{ours}
with the constant~$C$ of~\eqref{polynomial} 
due to \cite{Frank-Morozov-Vugalter_2011,Kovarik_2022,Baur},
these weakly coupled eigenvalues quantitatively match.  

The paper is organised as follows.
The singular Pauli operator~$H_{a_\alpha}$ 
formally written in~\eqref{AB-Hamiltonian}
is rigorously introduced in Section~\ref{Sec.operator}
via the method of self-adjoint extensions of symmetric operators.
Its Green function is determined in Section~\ref{Sec.Green}.
In Section~\ref{Sec.BS}, we start to develop 
the Birman--Schwinger analysis, which enables one to reduce
the study of eigenvalues of the differential operator~$H_{a_\alpha}$
subject to perturbations  to the study of an integral operator.
This is completed in Section~\ref{Sec.weak} 
by reducing everything to a matrix eigenvalue problem
and further to an implicit equation,
obtaining in this way the asymptotics of the weakly coupled 
eigenvalues.

\section{The Pauli operator via the extension theory}\label{Sec.operator}
%
The goal of this section is to rigorously introduce 
the singular Pauli operator~\eqref{AB-Hamiltonian}
and state its basic criticality properties. 
From now on, we abbreviate $H_\alpha := H_{a_\alpha}$.

Recall that the magnetic field associated with
the Aharonov--Bohm potential~\eqref{potential} is the distribution
$
  b_\alpha := \curl a_\alpha 
  = \partial_{x^1} a_\alpha^2 - \partial_{x^2} a_\alpha^1
  = -2\pi \alpha \delta
$.
It is therefore natural to introduce the operator~\eqref{AB-Hamiltonian} 
via the methods of extension theory of symmetric operators.
We follow the approach of
\cite{Geyler-Stovicek_2004_two_fluxes,Geyler-Stovicek_2004},
which is based on the factorisations
\begin{equation*}
  H^+ = T_- T_+ 
  \qquad \mbox{and} \qquad
  H^- = T_+ T_- 
  \,,
\end{equation*}
considered as operator identities in 
$
  \mathcal{D}(\Real^2 \setminus \{0\}) :=
  C_0^\infty(\Real^2 \setminus \{0\})
$,
where
\begin{equation*}
  T_\pm \coloneqq -i (\partial_{x^{1}} \pm i\partial_{x^{2}}) 
  - (a_\alpha^{1} \pm i a_\alpha^{2}) 
  =
  e^{\pm i \theta} 
  \left(
  -i\partial_r \pm i \frac{\alpha}{r} \pm \frac{\partial_\theta}{r}
  \right)
  \,.
\end{equation*}
Here the second equality follows by the usage of polar coordinates 
$(r, \theta)\in [0, \infty)\times (-\pi, \pi)$
defined by $(x^{1},x^2) =: (r \cos \theta, r \sin \theta)$.
We use the same symbols~$T_\pm$ for the extensions 
(denoted by $\tilde{T}_\pm$ in~\cite{Geyler-Stovicek_2004}) 
to the space of distributions $\mathcal{D}'(\Real^2 \setminus \{0\})$.

The \emph{minimal} realisations
$H^\pm_\mathrm{min} \coloneqq T_\pm^* \bar{T}_\pm$
are associated with the closure of the quadratic forms
\begin{equation*}
  h^\pm_\mathrm{min}[\psi]
  \coloneqq \|T_\pm\psi\|^2
  \,, \qquad 
  \Dom(h^\pm_\mathrm{min}) 
  \coloneqq  C_0^\infty(\Real^2 \setminus \{0\})
  \,.
\end{equation*}
It turns out (\cf~\cite[Lem.~5.2]{Geyler-Stovicek_2004}) 
that $H^+_\mathrm{min} = H^-_\mathrm{min} = -\Delta_{a_\alpha}$,
where~$-\Delta_{a_\alpha}$ is the usual magnetic Laplacian 
with the Aharonov--Bohm potential
(\ie~the Friedrichs extension of this operator
initially defined on $C_0^\infty(\Real^2 \setminus \{0\})$,
see~\cite{K7}).  
It is well known \cite{Laptev-Weidl_1999}
that~~$-\Delta_{a_\alpha}$ is subcritical 
(\ie~there is a Hardy-type inequality) whenever $\alpha \not\in \Int$.

The object of our interest are the \emph{maximal} realisations
$H^\pm_\mathrm{max} \coloneqq \bar{T}_\mp T_\mp^*$
associated with the closure of the quadratic forms
(\cf~\cite[Lem.~5.4]{Geyler-Stovicek_2004})
\begin{equation*}
  h^\pm_\mathrm{max}[\psi]
  \coloneqq \|T_\pm\psi\|^2
  \,, \qquad 
  \Dom(h^\pm_\mathrm{max}) 
  \coloneqq \{\psi \in \sii(\Real^2) : \ T_\pm\psi \in \sii(\Real^2)\}
  \,.
\end{equation*}
We then define 
\begin{equation}\label{direct}
  H_\alpha := H^+_\mathrm{max} \oplus H^-_\mathrm{max}
  \qquad \mbox{in} \qquad
  \sii(\Real^2,\Com^2) \cong \sii(\Real^2) \oplus \sii(\Real^2)
  \,.
\end{equation}

Contrary to the case of regular magnetic fields~\cite{Kovarik_2022}
(when the Pauli operator is essentially self-adjoint),
it turns out that \emph{both} $H^\pm_\mathrm{max}$ are critical
(\ie~they satisfy no Hardy-type inequality).
To see it, let us restrict from now on, without loss of generality
(by a gauge invariance, see \cite[Sec.~6]{Geyler-Stovicek_2004} 
and references therein or \cite[Prop.~3.1]{Per05}), 
to 
\begin{equation*}
  \alpha \in (0,1) 
  \,.
\end{equation*}
Then the virtual bound states~$\phi^\pm$ of~$H^\pm_\mathrm{max}$
are given by 
(note that~$\phi^-$ is related to~$\Omega_0^-$ of
\cite{Frank-Morozov-Vugalter_2011})
\begin{equation}\label{eq:virtual_states}
  \phi^-(x) \coloneqq  
    r^{-\alpha} 
  \qquad \mbox{and} \qquad
  \phi^+(x) \coloneqq 
    r^{-(1-\alpha)} e^{-i \theta} \,.
\end{equation}
More specifically, it is easy to check that 
$\phi^\pm \in \sii_\mathrm{loc}(\Real^2)$
and $T_\pm \phi^\pm = 0$ in the sense of distributions.
Of course, $\phi^\pm \not\in \sii(\Real^2)$,
however, the boundedness of~$\phi^\pm$ off the origin
enables one to apply the usual approximation procedure.

\begin{Lemma}\label{Lem.convergence}
There exists a sequence 
$\{\phi_n^\pm\}_{n \in \Nat} \subset \Dom(h^\pm_\mathrm{max})$
converging to $\phi^\pm$ pointwise and satisfying 
\begin{equation*}
  \lim_{n\to\infty} h^\pm_\mathrm{max}[\phi_n^\pm] = 0
  \,.
\end{equation*}
\end{Lemma}
\begin{proof}
Let us consider a radial function $\xi \in C^\infty(\Real^2)$
satisfying $0 \leq \xi \leq 1$,
$\xi(r) = 1$ if $0\leq r\leq 1$ 
and $\xi(r) = 0$ if $r\geq 2$.
Here, with an abuse of notation, we write $\xi(r) = \xi(x)$ when $|x|=r$.
Set $\xi_n(r) := \xi(r/n)$ for every $n > 0$.
Since,
$T^+(\phi^+ \xi_n) = r^{\alpha-1} (-\ii \p_r\xi_n)$
with $\phi^{+}$ from \eqref{eq:virtual_states},  
we have
\begin{align*}
  h_\mathrm{max}^+ [\phi^+ \xi_n]
 		= 2\pi \int_{n}^{2n} r^{2\alpha-2}|\xi'_n(r)|^2r\dd r
 		= \frac{2\pi}{n^2} 
 		 \int_{n}^{2n} r^{2\alpha-1}|\xi'(r/n)|^2\dd r 
 		 \leq 
  \pi \, \|\xi'\|_\infty^2 \frac{4^{\alpha}-1}{\alpha} \, n^{2\alpha-2}	
 		 \,.
\end{align*}
Hence 
$h_\mathrm{max}^+ [\phi^+ \xi_n] \to 0$ as $n \to \infty$
whenever $\alpha \in (0,1)$.
Similarly, we have
$T^-(\phi^-\xi_n) = \ee^{-\ii \theta}r^{-\alpha}(-\ii \p_r \xi_n)$
and thus
\begin{align*}
  h_\mathrm{max}^- [\phi^- \xi_n]
 		= 2\pi \int_{n}^{2n} r^{-2\alpha}|\xi'_n(r)|^2r\dd r
 		= \frac{2\pi}{n^2} 
 		 \int_{n}^{2n} r^{-2\alpha+1}|\xi'(r/n)|^2\dd r 
 		 \leq 
  \pi \, \|\xi'\|_\infty^2 \frac{4^{1-\alpha}-1}{1-\alpha} 
  \, n^{-2\alpha}	
 		 \,.
\end{align*}
Hence again
$h_\mathrm{max}^- [\phi^- \xi_n] \to 0$ as $n \to \infty$
whenever $\alpha \in (0,1)$.
\end{proof}

With this auxiliary lemma, we now establish the desired claim.

\begin{Proposition}
The operators $H^\pm_\mathrm{max}$ are critical.
\end{Proposition}
\begin{proof}
By contradiction, let us assume that there exists a non-trivial
non-negative function $\rho^\pm \in L^1_\mathrm{loc} (\Real^2)$ such that  
$H^\pm_\mathrm{max} \geq \rho^\pm$ 
in the sense of forms in $\sii(\Real^2)$.
Then, for any compact set $K \subset \Real^2 \setminus\{0\}$,
Lemma~\ref{Lem.convergence} implies
$\int_K \rho^\pm |\phi^\pm|^2 = 0$.
Since~$|\phi^\pm|$ are positive on arbitrary~$K$,
we conclude with the contradiction that $\rho^\pm=0$
almost everywhere in~$\Real^2$.
\end{proof}

Let us now characterise the domain of the Pauli operator~$H_\alpha$
via boundary conditions at the singularity $r=0$.
The action of the operators $H^{\pm}_{\mathrm{max}}$ 
and $-\Delta_{a_\alpha}$ coincide 
on $C_0^{\infty}(\Rb^2 \setminus \{0\})$.
Hence they are two different
self-adjoint extensions of the symmetric operator 
\begin{align}\label{eq:X_tilde}
	\tilde{X} 
		:=  -\Delta_{a_\alpha}
		= -\p_r^2 - r^{-1}\p_r +r^{-2}(-\ii \p_{\theta} + \alpha)^2 \,,
		\qquad
		\Dom(\tilde{X}) \coloneqq C_0^{\infty}(\Rb^2 \setminus \{0\}) \,.
\end{align}
Due to~\cite{Geyler-Stovicek_2004_two_fluxes},
one has the characterisation
\begin{subequations}\label{eq:domPauli}
	\begin{align}
		\label{eq:domH+}
		\Dom(H^+_{\mathrm{max}})
	 		&= 	\{ f \in \Dom(\tilde{X}^{\ast}) : \
	 			 		\Phi_2^{-1} (f) = \Phi_1^0(f) = 0
	 		\} \,, 
		\\
		\label{eq:domH-}
	 	\Dom(H^-_{\mathrm{max}})
	 		&= 	\{ f \in \Dom(\tilde{X}^{\ast}) : \ 
	 					\Phi_1^{-1} (f) = \Phi_2^0(f) = 0
	 		\} \,,
	\end{align}
\end{subequations}
where
\begin{align*}
	\Phi_1^{-1}(f) 
		&\coloneqq \lim_{r \rightarrow 0} r^{1-\alpha} \frac{1}{2\pi} \int_0^{2\pi} f(r, \theta) \ee^{\ii \theta} \dd \theta \,,	\\
	\Phi_2^{-1}(f) 
		&\coloneqq \lim_{r \rightarrow 0} r^{-1+\alpha} 
		\left(\frac{1}{2\pi} \int_0^{2\pi} f(r, \theta) \ee^{\ii \theta} \dd \theta - r^{-1+\alpha} \Phi_1^{-1}(f) \right)\,,	\\
	\Phi_1^{0}(f) 
		&\coloneqq \lim_{r \rightarrow 0} r^{\alpha} \frac{1}{2\pi} \int_0^{2\pi} f(r, \theta) \dd \theta \,,	\\
	\Phi_2^{0}(f) 
		&\coloneqq \lim_{r \rightarrow 0} r^{-\alpha} \left( \frac{1}{2\pi} \int_0^{2\pi} f(r, \theta) \dd \theta -r^{-\alpha} \Phi_1^0(f) \right) 
	 \,.
\end{align*}

It is not difficult to find the spectrum of the Pauli operator.
\begin{Proposition}\label{Prop.spectrum}
$
  \sigma(H_\alpha) = \sigma_\mathrm{ess}(H_\alpha) = [0,+\infty)
$.	 
\end{Proposition}
\begin{proof}
Since~$H^{\pm}_{\mathrm{max}}$ are non-negative,
we immediately have $\sigma(H_\alpha) \subset [0,+\infty)$.
To prove the opposite inclusion, we construct a Weyl sequence.
Given $\phi\in C_0^{\infty}((0,\infty) \times \Real)$ 
with $\|\varphi\|=1$,
for every positive~$n$ we define 
$$
  \phi_n(x) := \frac{1}{n} \,
  \phi\left( \frac{x^1}{n}-n,\frac{x^2}{n} \right)
  .
$$
Note that the normalisation factor is chosen in such a way
that $\|\phi_n\|=1$ for every~$n$.
Moreover, the scaling ensures that the derivatives of~$\varphi_n$ 
vanish as $n \to \infty$, namely 
\begin{equation}\label{scaling}
  \|\nabla\varphi_n\| = n^{-1} \, \|\nabla\varphi\|
  \qquad\mbox{and}\qquad
  \|\Delta\varphi_n\| = n^{-2} \, \|\Delta\varphi\|
  \,.
\end{equation}
Finally, the shift guarantees that the support of~$\varphi_n$
never intersects the origin where the operator is singular,
in fact the support is ``localised at infinity'' in the sense
that $\supp\varphi_n = (n^2,0) + n \supp\varphi $.  
Now we define 
$\psi_n(x) := \phi_n(x) \ee^{\ii k \cdot x}$,
where $k \in \Real^2$. 
Note that 
$\psi_n\subset C_0^{\infty}(\Rb^2\setminus \{0\})
\subset \Dom (H^{\pm}_{\mathrm{max}})$
and $\|\psi_n\|=1$ for every~$n$, 
while both $H^{\pm}_{\mathrm{max}}$ act on $C_0^{\infty}(\Rb^2\setminus \{0\})$ 
 	as $\tilde{X}$ introduced in~\eqref{eq:X_tilde}.
Using that~$a_\alpha$ is divergence-free outside the origin,
one has 	
 	\begin{align*}
 		\|(\tilde{X} -k^2) \psi_n\| 
 			\leq 
 			\|\Delta \phi_n\| 
 			+ 2 |k| \|\nabla \phi_n\| + \||a_\alpha|^2 \phi_n\|
 			+ 2 \| a_\alpha\cdot \nabla \phi_n\|
 	\xrightarrow[n \to \infty]{} 0 \,.
	\end{align*} 
Here, in addition to~\eqref{scaling}, we have used that 
$\|a_\alpha\|_{L^\infty(\supp\varphi_n)} \to 0$ as $n \to \infty$.
This argument shows that 
$
  \sigma(H^{\pm}_{\mathrm{max}}) 
  = \sigma_\mathrm{ess}(H^{\pm}_{\mathrm{max}}) = [0,+\infty)
$.	 	
	The same spectral result for~$H_\alpha$
	now follows by the fact that the spectrum of $H_\alpha$ 
	is the union of the spectra of $H^+_{\mathrm{max}}$
 	and $H^-_{\mathrm{max}}$ due to~\eqref{direct}.
\end{proof}

In parallel with the Friedrichs extension $-\Delta_{a_\alpha}$ of~$\tilde{X}$,
the authors of \cite{Sto89,Geyler-Stovicek_2004_two_fluxes}
consider the unitarily equivalent operator 
\begin{equation}\label{eq:Unitary}
	H  \coloneqq U_{\alpha}(-\Delta_{a_\alpha})U_{\alpha}^{-1}
  \,, \qquad \mbox{where} \qquad
 	U_{\alpha}\phi(r, \theta) 
 	\coloneqq \ee^{\ii \alpha \theta} \phi(r, \theta)
 	\,.
\end{equation}
The operator~$H$ acts as the Laplacian in~$\Real^2 \setminus \{0\}$
and functions~$\psi$ in its domain satisfy the following 
boundary conditions at the origin and on the cut $\theta = \pi$: 
\begin{subequations}
\label{eq:Friedrichs}
\begin{align} 
	\psi(0) 		&= 0 \,, 	\\
 	\psi(r, \pi) 	&= \ee^{2\pi \ii \alpha} \psi(r, -\pi) \,,	\\
 	\p_r\psi(r, \pi) 	&= \ee^{2\pi \ii \alpha} \p_r \psi(r, -\pi)	\,.
\end{align}
\end{subequations}
%

\section{The Green function}\label{Sec.Green}
%
The goal of this section is two-fold. 
First, we apply the Krein's formula 
to the Green function~$G_z$ of the operator~$H$ from~\eqref{eq:Unitary}
to find the Green function of the Pauli operator~$H_{a_\alpha}$.
Second, we study the singularities of the Green function.

\subsection{Krein's formula}
The Green function~$G_z$ is presented in \cite[Eq.~(7)]{Sto89}.
Denoting by $z\in \Cb \setminus [0, \infty)$ the spectral parameter 
 and choosing the branch of the square root so that
 $\Re \sqrt{-z}>0$ and recalling that we assume $\alpha \in (0,1)$
 (without loss of generality), 
 it reads
\begin{subequations} \label{eq:plain_Green_fun}
	\begin{align} 
		G_z(r, \theta; r_0, \theta_0)
			&=	\hat{C}(\theta - \theta_0) \,
				K_0\big(\sqrt{-z} \, |x-x_0|\big) 
				\label{eq:plain_Green_fun_a}
				\\
			& \quad -\frac{\sin(\pi \alpha)}{\pi} 
				\int_{-\infty}^{\infty} \frac{1}{2\pi} \,
				K_0\big(\sqrt{-z} \, R(s)\big) \,
				\frac{\ee^{-\alpha s + \ii \alpha (\theta - \theta_0)}}{1 + \ee^{-s + \ii (\theta - \theta_0)}} \dd s \,.
				\label{eq:plain_Green_fun_b}
	\end{align}
\end{subequations}
Here $K_0$ is the zero-th modified Bessel function of the second kind, 
\begin{align} \label{eq:R}
	|x-x_0|^2 = r^2 + r_0^2 - 2r r_0 \cos (\theta- \theta_0) \,, 
 	\quad R(s)^2 := r^2 + r_0^2 +2rr_0\cosh (s) \,,
\end{align}
and
\begin{equation}\label{tilda}
  \hat{C}(\theta - \theta_0) 
  := \frac{1}{2\pi}		 
				\begin{cases}
					1 
					& \mbox{if}\quad 
					\theta - \theta_0 \in (-\pi, \pi) \,,	
					\\
					\ee^{-2\pi \ii \alpha} 	
					& \mbox{if}\quad	
					\theta - \theta_0 \in (-2\pi, -\pi) \,,
					\\
					\ee^{2\pi \ii \alpha}
					& \mbox{if}\quad	
					\theta - \theta_0 \in (\pi, 2\pi) \,.
				\end{cases}
\end{equation}
\begin{Remark}\label{Rem.three}
Despite the three-fold description, 
$G_z$~is continuous at $\theta-\theta_0=\pm \pi$.
Indeed, the continuity of the first line~\eqref{eq:plain_Green_fun_a}
follows from the equality \cite[\S\,6.791-1]{GR07} 
 	\begin{align*}
 		\pi K_0(a+b) = \int_{-\infty}^{\infty} K_{\ii \tau}(a)K_{\ii \tau}(b) \dd \tau
 	\end{align*}
 	for $|\ph(a)|+|\ph(b)|\leq \pi$.
To see the continuity of the second line~\eqref{eq:plain_Green_fun_b}, 
we use the identity 	
 	\begin{align} \label{eq:continuity_in_theta}
 		\int_{-\infty}^{\infty} K_{\ii \tau}(a) K_{-\ii \tau}(b) 
 			\frac{\ee^{\phi \tau}}{\sin(\pi(\alpha +\ii \tau))} \dd \tau
 		=
 		\int_{-\infty}^{\infty} K_0(\sqrt{a^2+b^2 +2ab \cosh (u)})
 			\frac{\ee^{-\alpha(u-\ii \phi)}}{1+\ee^{-u+\ii \phi}} \dd u
 		\,,
 	\end{align}
 	for $\ph (a), \ph(b)<\pi$, $\alpha \in (0, 1)$ and $|\phi| <\pi$. 	
 	Its validity can be checked from formula
 	\cite[\S\,6.792-2]{GR07} 
	\begin{align} \label{eq:continuity_2}
	 	\int_{-\infty}^{\infty} \ee^{\ii u\tau} K_{\ii \tau}(a) K_{\ii \tau}(b) \dd \tau
	 		= 
	 		\pi K_0(\sqrt{a^2+b^2 +2ab \cosh (u)}) \,.
	\end{align}
	Using the fact that for $|\phi|<\pi$ the integral
	$\oint_{\gamma}
 			\ee^{\ii u\tau} 
 			\frac{\ee^{-\alpha(u-\ii \phi)}}{1+\ee^{-u+\ii \phi}} 
 			\dd u $
 	vanishes along the rectangle
 	$\gamma := (-R, R)\cup (R, R-\ii \phi) \cup (R-\ii \phi, -R - \ii \phi)
 	\cup (-R-\ii \phi, -R)$	for any $R>0$
	 by the residue theorem, we further compute
	\begin{align*}
	 	\int_{-\infty}^{\infty} \ee^{\ii u \tau} \frac{\ee^{-\alpha(u-\ii \phi)}}{1+\ee^{-u+\ii \phi}} \dd u
	 		= 
	 		\int_{-\infty}^{\infty} \ee^{\ii \tau (s+\ii \phi)} \frac{\ee^{-\alpha s}}{1+\ee^{-s}} \dd s
	 		= \ee^{-\tau \phi} B(\alpha-\ii \tau, 1-(\alpha-\ii \tau))
	 	 \,.
	\end{align*}
	For the last equality we have used~\cite[\S\,3.313-2]{GR07}.
	On the right-hand side the Beta function~$B$ can be further evaluated as
	$B(\alpha-\ii \tau, 1-(\alpha-\ii \tau))= \frac{\pi}{\sin(\pi(\alpha-\ii \tau))}$.
	Now we apply this result to the left-hand side of the  
	$\frac{\ee^{-\alpha(u-\ii \phi)}}{1+\ee^{-u+\ii \phi}}$
	multiple of \eqref{eq:continuity_2} integrated over 
	$u\in (-\infty, \infty)$ and arrive at 
	\eqref{eq:continuity_in_theta}.
\end{Remark}

To obtain the Green functions $G_z^{\pm}$ of the extensions
$H^{\pm}_{\mathrm{max}}$, we mimic the steps 
in~\cite[Sec.~IV]{Geyler-Stovicek_2004_two_fluxes} 
and use the Krein's formula (\cf~\cite[Eq.~(2.6)]{DG85}).
Recalling the unitary transformation $U_\alpha$ 
from~\eqref{eq:Unitary},  the Krein's formula yields
\begin{equation}\label{Krein}
 	 \ee^{\ii \alpha \theta} \,
 	 {G}_z^{\pm} (x,x_0) \,
 	 \ee^{-\ii \alpha \theta_0}
 	  		={G}_z(x,x_0)
 		+\sum_{j,k=1,2} (M_z^{\pm})^{j,k} 
 						f_z^j(x) \overline{f_{\bar{z}}^k(x_0)} 
 		\,.
\end{equation}
The coefficient matrices $M_z^{\pm}$ are determined below.  
 The functions $f_z^{1},f_z^{2}$ form a basis of the 
deficiency subspaces $\text{ker } (X^{\ast}-z)$, 
where $X$ is the Laplacian on test functions on $\Rb^2\setminus \{0\}$.
In particular, we set (\cf~\cite{DS98})
	\begin{align*}
		\{f_z^1(r, \theta), f_z^2(r, \theta)\} 
			\coloneqq 
			\{ K_{1-\alpha}(\sqrt{-z}r) \ee^{\ii (\alpha-1) \theta} ,
				K_{\alpha}(\sqrt{-z}r) \ee^{\ii \alpha\theta} 
				\}\,,
	\end{align*}
with $K_{\nu}$ denoting the $\nu$-th modified Bessel function of the second kind.

As~$G_z$ is the integral kernel of the resolvent of~$H$, 
the range of the corresponding integral operator is~$\Dom(H)$ 
determined by the boundary conditions~\eqref{eq:Friedrichs}. 
The sum in~\eqref{Krein} 
is the integral kernel of a finite-rank operator. 
Since upon integration over 
$x_0\in \Rb^2$
the Green function $G_z^{\pm}(x, x_0)$ has to map square integrable functions to the domain 
$\Dom (H^{\pm}_{\mathrm{max}})$, we need to choose
the matrices $M_z^{\pm}$ in such a manner that 
$ G_z^{\pm}(x,x_0)$ 
satisfy (as functions of $x\in \Rb^2$) the respective boundary conditions of $H^{\pm}_{\mathrm{max}}$ 
given in~\eqref{eq:domPauli}.
To check these conditions, 
we investigate the behaviour of ${G}_z(x, x_0)$  
and that of $f_z^1,f_z^2$ for $|x|=r\rightarrow 0$.
As for the former, one has \cite[Eq.~(29)]{Geyler-Stovicek_2004_two_fluxes}
\begin{multline*}
	{G}_z(r, \theta, r_0, \theta_0) 
		=\frac{\sin(\pi \alpha)}{2\pi^2} \frac{\Gamma(\alpha)}{1-\alpha}
		\left (
			\frac{\sqrt{-z}r}{2}
		\right )^{1-\alpha}
		\overline{f^1_{\bar{z}}(r_0, \theta_0)} \ee^{-\ii(1-\alpha) \theta} 
	\\
		+
		\frac{\sin(\pi \alpha)}{2\pi^2} \frac{\Gamma(1-\alpha)}{\alpha}
		\left (
			\frac{\sqrt{-z}r}{2}
		\right )^{\alpha}
		\overline{f^2_{\bar{z}}(r_0, \theta_0)} \ee^{\ii\alpha \theta} 
		+
		\mathcal{O}(r)	 
\end{multline*}
as $r\rightarrow 0$. 
The asymptotics of $f_z^1,f_z^2$ follow 
from the behaviour of the Bessel functions 
\begin{align}\label{eq:Bessel_asy_small}
 	K_{\nu}(w) 
 		= 
 		\frac{\Gamma(\nu)}{2} 
 			\left( \frac{w}{2} \right) ^{-\nu} 
 			(1+\mathcal{O}(w^2))
 		-\frac{\Gamma(1-\nu)}{2\nu} 
 			\left( \frac{w}{2}\right) ^{\nu} 
 			(1+\mathcal{O}(w^2)) 
 			\,,
\end{align}
as $|w|\rightarrow 0$, 
where $\Re(\nu) > 0$ and $\ph(w)\neq \pm \pi$.
In particular, we arrive at
\begin{align*}
\lefteqn{
 	{G}^{\pm}_z(r, \theta, r_0, \theta_0) \, \ee^{-\ii \alpha \theta_0}
 	}
 	\\
 		\qquad &=
 			\sum_{j=1,2} (M^{\pm}_z)^{j,1} 
 				\overline{f_{\bar{z}}^j(r_0, \theta_0) }
 				\bigg[
 					\frac{\Gamma(1-\alpha)}{2} 
 					\left( \frac{r}{2}\right)^{\alpha-1}
 					(\sqrt{-{z}})^{\alpha-1} 
 					-
 					\frac{\Gamma(\alpha)}{2(1-\alpha)}
 					\left( \frac{r}{2}\right)^{1-\alpha}
 					(\sqrt{-{z}})^{1-\alpha}  
 				\bigg] 
 					\ee^{-\ii \theta} 	
 				\\
 			&\quad
 				+ 
 				 \sum_{j=1,2} (M^{\pm}_z)^{j,2} 
 				\overline{f_{\bar{z}}^j(r, \theta)} 
 				\bigg[
 					\frac{\Gamma(\alpha)}{2} 
 					\left( \frac{r}{2}\right)^{-\alpha} 
 					(\sqrt{-z})^{-\alpha}
 					- 
 					\frac{\Gamma(1-\alpha)}{2\alpha} 
 					\left( \frac{r}{2}\right)^{\alpha} 
 					(\sqrt{-{z}})^{\alpha} 
 				\bigg] 	\\
 			&\quad
 				+ 
 				\frac{\sin(\pi \alpha)}{2\pi^2} 
 				\frac{\Gamma(\alpha)}{1-\alpha} 
 				\left(\frac{r}{2}\right)^{1-\alpha} 
 				(\sqrt{-z})^{1-\alpha}
 				\ee^{-\ii \theta} 
 				\overline{f_{\bar{z}}^1(r_0, \theta_0)}	\\
			&\quad 			 
 			 	+ 
 			 	\frac{\sin(\pi \alpha)}{2\pi^2} 
 			 	\frac{\Gamma(1-\alpha)}{\alpha} 
 				\left(\frac{r}{2}\right)^{\alpha} 
 				(\sqrt{-{z}})^{\alpha}
 				\overline{f_{\bar{z}}^2(r_0, \theta_0)}
 				+
 				\mathcal{O}(r) 
\end{align*}
as $r \rightarrow 0$.

Now we are able to check when
${G}_z^{\pm}(r, \theta, r_0,\theta_0)$
satisfy the boundary conditions~\eqref{eq:domPauli}.
We first consider $\Dom(H^{+}_{\mathrm{max}})$.
The need of the functional $\Phi_1^0(\cdot)$ to vanish means that the
coefficients in front of 
$r^{-\alpha}$ have to vanish.
This leads to 
 $(M^{+}_z)^{j,2} = 0$ for all $j=1,2$ (since $f_z^{1,2}$ are linearly independent).
From vanishing of the functional $\Phi_{2}^{-1}(\cdot)$ 
we conclude that the coefficient in front of the term 
$ r^{1-\alpha} \ee^{-\ii \theta} $ is zero.
Therefore
\begin{align*}
	(M_z^+)^{2,1} = 0
	\quad
	\mathrm{ and }
	\quad
	(M_z^+)^{1,1} 
		= \frac{\sin (\pi \alpha)}{\pi^2} \,.
\end{align*}
Secondly, we deal with the conditions in $\Dom (H^-_{\mathrm{max}})$ and obtain
\begin{align*}
 	(M_z^-)^{j,1}= 0 \,, \quad j = 1,2 \,,
\end{align*}
by vanishing coefficients 
in front of $ r^{\alpha-1} \ee^{-\ii \theta}$ coming 
from the requirement $\Phi_1^{-1}(\cdot) = 0$,
and, finally, 
\begin{align*}
 	(M_z^-)^{1,2}= 0 
 	\quad
	\mathrm{ and }
	\quad
	(M_z^-)^{2,2} 
 		= \frac{\sin (\pi \alpha)}{\pi^2}
\end{align*}
as $\Phi_2^{0}(\cdot)=0$ implies
vanishing coefficients 
in front of $r^{\alpha} $.

We summarise our investigation in the following proposition.
\begin{Proposition} \label{prop:Green_fun_Pauli}
Let $\alpha \in (0,1)$.
For every $z \in \Com\setminus [0,+\infty)$,
the resolvent of~$H_\alpha$ satisfies 
$$
		(H_\alpha-z)^{-1}
		= 
		\begin{pmatrix}
				(H_\mathrm{max}^+ -z)^{-1}	& 0 	\\
				0 		& (H_\mathrm{max}^- -z)^{-1}
		\end{pmatrix}
		,
$$
where the integral kernels~$G_z^\pm$ of the operators 
$(H_\mathrm{max}^\pm -z)^{-1}$ are given by 
	\begin{align*}
		{G}_z^+(r, \theta, r_0, \theta_0)
			&= 
				\ee^{-\ii \alpha \theta} \,
				{G}_z (r, \theta, r_0, \theta_0) \,
				\ee^{\ii \alpha \theta_0} 
				+
			\frac{\sin (\pi \alpha)}{\pi^2}  K_{1-\alpha} (\sqrt{-z}r) 
			\overline{K_{1-\alpha} (\sqrt{-\bar{z}}r_0)}
			\ee^{-\ii (\theta - \theta_0)} 
					\\
		{G}_z^-(r, \theta, r_0, \theta_0)
			&= 
				 \ee^{-\ii \alpha \theta} \,
				{G}_z (r, \theta, r_0, \theta_0) \,
				\ee^{\ii \alpha \theta_0}
			  + 
			\frac{\sin (\pi \alpha)}{\pi^2}  K_{\alpha} (\sqrt{-z}r) 
			\overline{K_{\alpha} (\sqrt{-\bar{z}}r_0)}
				\,,
	\end{align*}
with~$G_z$ being given in~\eqref{eq:plain_Green_fun}.
\end{Proposition}	

\subsection{The singularities}
It is well known that the criticality of an operator
is related to the singularity of its Green function.
Moreover, the weak-coupling asymptotics are determined
by the nature of the singularity. 
To apply the Birman--Schwinger analysis below,
we need to have precise information about the singularities 
of the Green functions ${G}_z^\pm$ as $z \to 0$.

Note that the functions $z \mapsto G_z^{\pm}$ 
are analytic outside of the cut $[0,+\infty)$. 
Moreover, the limits $G_{k+\ii \eps}^{\pm}$ as $\eps \to 0^\pm$
are well defined for every positive~$k$.
Therefore the singularities indeed occur only as~$|z|$ approaches zero.
 
First of all, we establish a technical identity. 
\begin{Lemma}\label{Lem.residual} 
Let $\alpha \in (0,1)$.
For every $\varphi \in (-2\pi,-\pi) \cup (-\pi,\pi) \cup (\pi,2\pi)$,
one has 
\begin{equation}\label{eq:singularity}
 2\pi \hat{C}(\phi)
				-\frac{\sin(\pi \alpha)}{\pi}
				\ee^{\ii \alpha \phi}
		\int_{-\infty}^{\infty} 
				\frac{\ee^{-\alpha s}}{1 + \ee^{-s + \ii \phi}} \dd s 
		= 0 \,,
\end{equation}
where~$\hat{C}(\phi)$ is given in~\eqref{tilda}.  
\end{Lemma}
\begin{proof}
Note that the integral converges 
under the restrictions on~$\phi$ and~$\alpha$.
To compute its value we use the residue theorem. To that aim
we consider the regularised integrand of \eqref{eq:singularity}
by multiplying it by $\ee^{-\ii \eps s}$ for some $\eps>0$.
Then we integrate along the oriented contour~$C_R$
consisting of the real interval $[-R, R]$
and the arc~$\Gamma_R$ of radius~$R$ 
centred at the origin placed in the top half of the complex plane. 
In summary, we are considering
\begin{align*}
  J_\eps :=
	\oint_{C_R}   \ee^{\ii \eps s}	\frac{\ee^{-\alpha s}}{1 + \ee^{-s + \ii \phi}} \dd s \,.
\end{align*}
This integral can be evaluated by summing its residua 
$\{\exp[-(\eps+\ii \alpha)(\phi +(2k+1)\pi)] \}_{k=k_\mathrm{min}}^{\infty}$ 
corresponding to the simple poles $s_k := \ii (\phi + (2k+1)\pi) $ in the upper half of the complex plane as follows:
\begin{equation*}
	J_\eps
	=
	\ee^{-(\epsilon + \ii \alpha)(\phi+\pi)} 	
 	\sum_{k=k_\mathrm{min}}^{\infty} 
 		\ee^{-(\epsilon + \ii \alpha)2\pi k } 	
 		 \qquad \text{with} \qquad
 		 k_\mathrm{min}
 		 :=\begin{cases}
 		 	0 & \text{ if } \phi \in (-\pi, \pi) \,, \\
 		 	1 & \text{ if } \phi \in (-2\pi, -\pi) \,, \\
 		 	-1 & \text{ if } \phi \in (\pi, 2\pi) \,.
 		 \end{cases} 
\end{equation*}
Here $k_\mathrm{min}$ is given by the condition that $\Im(s_k)>0$ for all $k\geq k_\mathrm{min}$. Notice also that the sum is well defined as for all such $k$ 
the absolute value of the summands is less than one.
The integral over the arc of the loop vanishes in the limit $R\rightarrow \infty$ since
\begin{align*}
 	\left| 
 		\int_{\Gamma_R}  \ee^{\ii \eps s}	\frac{\ee^{-\alpha s }}{1 + \ee^{-s + \ii \phi}} \dd s
	\right|
		&\leq \int_{\Gamma_R} \left|   \ee^{\ii \eps s}	\frac{\ee^{-\alpha s }}{1 + \ee^{-s + \ii \phi}} \right| \dd s 	
		\\
		&=  \int_0^{\pi/2} \left|  \frac{\ee^{-\alpha R \cos \beta - \eps R \sin \beta }}{1 + \ee^{-R(\cos \beta + \ii \sin \beta) + \ii \phi}} \right| R \dd \beta 
		+  \int_{\pi/2}^{\pi} \left|  \frac{\ee^{(1-\alpha) R \cos \beta - \eps R \sin \beta }}{1 + \ee^{R(\cos \beta + \ii \sin \beta) - \ii \phi}} \right| R \dd \beta \,.
\end{align*}
Using the dominated convergence we can thus compute our original integral 
\begin{align*}
	\ee^{\ii \alpha \phi}
	\int_{-\infty}^{\infty} \frac{\ee^{-\alpha s }}{1 + \ee^{-s + \ii \phi}} \dd s	
		&= 	 \ee^{\ii \alpha \phi}\lim_{\eps \rightarrow 0} \int_{-\infty}^{\infty}   \ee^{\ii \eps s}	\frac{\ee^{-\alpha s }}{1 + \ee^{-s + \ii \phi}} \dd s
		\\
		&=	\ee^{\ii \alpha \phi}\lim_{\epsilon\rightarrow 0} \lim_{R \rightarrow \infty}
			\int_{C_R} \ee^{\ii \eps s}	\frac{\ee^{-\alpha s }}{1 + \ee^{-s + \ii \phi }} \dd s
		=
		\frac{\pi}{\sin (\pi \alpha)}
		\left.
 		\begin{cases}
 			1	\\
 			\ee^{-2\pi \ii \alpha} 	\\
 			\ee^{2\pi \ii \alpha}
		\end{cases}
		\right\}
		 \,.
\end{align*}
This is equivalent to~\eqref{eq:singularity}.
\end{proof}

To study the singularities of~$G_z^\pm$, 
we start with the part $G_z$ given by~\eqref{eq:plain_Green_fun}.
First, recall the asymptotics of the zero'th Bessel function 
(\cf~\cite[\S\,9.6.12--13]{AS64})
$K_0(w) = -\log (w/2) + \mathcal{O}(1)$ 
as $|w|\rightarrow 0$.
Second, since the function~$R$ introduced in~\eqref{eq:R}
is real-valued, we have
$
  \log(\sqrt{-z} R(s)) 
		= \log(\sqrt{-z}) + \log(R(s))
$.
Consequently, we see that the logarithmic singularity of~$G_z$ 
cancels out due to Lemma~\ref{Lem.residual},
yielding the behaviour
\begin{align*}
	G_z(r, \theta; r_0, \theta_0)
		&= d(x,x_0) + \mathcal{O}\big( (\sqrt{-z})^2 \big)
		\qquad \text{as} \qquad |z| \rightarrow 0 \,,
\end{align*}
with some function $d$ dependent on the spatial coordinates but independent of $z$.

Now we focus on the second terms in the formulae for $G_z^{\pm}$ 
given in  Proposition~\ref{prop:Green_fun_Pauli}.
From the asymptotics~\eqref{eq:Bessel_asy_small}, we compute
\begin{subequations}
\begin{align}
\begin{split} 	\label{eq:KK_1}
 	& K_{1-\alpha}(\sqrt{-z}r) \overline{K_{1-\alpha}(\sqrt{-\bar{z}}r_0)} 	\\
 		&\quad
 		= 
 		 (\sqrt{-z})^{2(\alpha-1)} \left(\frac{\Gamma(1-\alpha)}{2} \right)^2 
 		 	\left(\frac{r r_0}{4}\right)^{\alpha-1} 
 		 	- \frac{\Gamma(\alpha)\Gamma(1-\alpha)}{4(1-\alpha)} 
 		 		\left( \left(\frac{r_0}{r}\right)^{1-\alpha}
 		 			+\left(\frac{r}{r_0}\right)^{1-\alpha}\right) 	\\ 		
 		 &\qquad
 		 	+ \left( \frac{\Gamma(\alpha)}{2(1-\alpha)} \right)^2 (\sqrt{-z})^{2(1-\alpha)} 
 		 		\left( \left(\frac{r r_0}{4}\right)^{1-\alpha}
 		 			\right)  		 	
 		 	+ \mathcal{O}(\sqrt{-z}^{2 \min\{\alpha, 1-\alpha\} }) 
 		 	\,,  
 \end{split}
 	\\
\begin{split} 	\label{eq:KK_2} 	
 	& K_{\alpha}(\sqrt{-z}r) \overline{K_{\alpha}(\sqrt{-\bar{z}}r_0)}	
		\\ 		
 		&\quad
 		= 
 		 (\sqrt{-z})^{-2\alpha} \left(\frac{\Gamma(\alpha)}{2} \right)^2 
 		 	\left(\frac{r r_0}{4}\right)^{-\alpha} 
 		 	- \frac{\Gamma(\alpha) \Gamma(1-\alpha)}{4\alpha} 
 		 		\left( \left( \frac{r_0}{r} \right)^{\alpha}
 		 				+ \left(\frac{r}{r_0} \right)^{\alpha} \right)	\\		 	
 		 	&\qquad
 		 		+ \left( \frac{\Gamma(1-\alpha)}{2\alpha}\right)^2 \sqrt{-z}^{2\alpha} 
 		 				\left(\frac{r r_0}{4}\right)^{\alpha}
 		 	+ \mathcal{O}(\sqrt{-z}^{2\min\{\alpha, 1-\alpha\} }) \,,
\end{split}
\end{align}
\end{subequations}
as $|z| \rightarrow 0$. 

In summary, we have established the following asymptotics.
\begin{Proposition} 
Let $\alpha \in (0,1)$. For every $z \in \Com\setminus[0,+\infty)$,
one has
\begin{align*}
 	G_{z}^{+}(r, \theta, r_0, \theta_0)
 		&= 
 		C_{\alpha} \Gamma(1-\alpha)^2 (-z rr_0/4)^{\alpha-1} 
			 \ee^{-\ii (\theta - \theta_0)}	
 		 + \mathcal{O}(1) \,,
 		 \\
 	G_{z}^{-}(r, \theta, r_0, \theta_0)
 		&=
 		C_{\alpha} \Gamma(\alpha)^2 (-z rr_0/4)^{-\alpha} 
 		 + \mathcal{O}(1) \,,
\end{align*}  
as $|z| \rightarrow 0$,
where  
\begin{align*}
	C_{\alpha} 
		:= 
		\frac{\sin (\pi \alpha)}{4\pi^2} \,.
\end{align*}
\end{Proposition}
%

\section{The Birman--Schwinger analysis}\label{Sec.BS}
%
Let $V:\Real^2 \to \Com^{2 \times 2}$
be a matrix-valued function. 
For almost every $x \in \Real^2$, we consider 
the matrix polar decomposition
\begin{equation}\label{eq:AB}
  V(x) = B(x)A(x) 
  \qquad \mbox{with} \qquad
  \begin{aligned}
  A (x) &\coloneqq \sqrt[4]{V(x)^*V(x)}\,, \\
  B(x) &\coloneqq \mathbb{U}(x) \, \sqrt[4]{V(x)^*V(x)} \,,
  \end{aligned}
\end{equation}
where~$\mathbb{U}(x)$ is a unitary matrix.
By $|V(x)|$ we denote the operator norm of the matrix~$V(x)$ 
when considered as an operator on~$\Com^2$.
Our standing hypothesis about~$V$ is as follows.
\begin{Assumption} \label{as:potential}
Let $\alpha \in (0,1)$.
Suppose
	\begin{align} \label{eq:V_assumptions}
		\int_{\Rb^2} |V(x)| (|x|^{2\nu} + |x|^{-2\nu}) \dd x < \infty
		\qquad 
		\text{ and }
		\qquad |V| \in L^{1+\delta}(\Rb^2)
	\end{align}
	for some positive~$\delta$
	 and  
	 $\nu:= \max \{1-\alpha, \alpha\} $.	
\end{Assumption}

	Note that conditions~\eqref{eq:V_assumptions} particularly 
	imply that $|V| \in L^1(\Rb^2)$. 

Let~$\epsilon$ be a (small) positive number.
By the Birman--Schwinger principle~\cite{HK22}
(justified under Assumption~\ref{as:potential}
in Remark~\ref{Rem.BS} below),
$z \in \Com\setminus [0,+\infty)$ is an eigenvalue 
of $H_\alpha + \epsilon V$ if, and only if, 
$-1$ is an eigenvalue of the integral operator
\begin{equation}\label{BS}
  R_{z, \epsilon}  \coloneqq
				\epsilon A(H_\alpha -z)^{-1} B \,.
\end{equation}
Here $A,B$ are considered as the maximal operators of multiplication
by the matrix-valued functions denoted by the same symbols.  

To apply this principle to the analysis of the weakly coupled eigenvalues, 
we decompose
\begin{equation}\label{BS.decomposed}
  R_{z, \epsilon} = \epsilon (L_z + Q_z)
  \,,
\end{equation}
where~$L_z$ and~$Q_z$ are the singular and regular parts
of the Birman--Schwinger operator, respectively. 
More specifically, adopting the convention that the kernel
of an integral operator~$T$ is distinguished calligraphically by~$\mathcal{T}$, 
we define
\begin{subequations} \label{eq:Green_fun_split}
	\begin{align}
		\label{eq:L_z}
	 	\mathcal{L}_{z} (r, \theta, r_0, \theta_0)
	 		\coloneqq & \
	 		C_{\alpha}
	 		A(r, \theta)
	 		\begin{pmatrix}
	 	 		 \Gamma(1-\alpha)^2
	 	 		 	(-z |rr_0|/4)^{\alpha-1} \ee^{-\ii(\theta-\theta_0)}
	 		 & 	0
	 		 \\
	 		 0 	&
	 		 	  \Gamma(\alpha)^2 
	 		 	  	(-z |rr_0|/4)^{-\alpha}
	 		\end{pmatrix} 
	 		B(r_0, \theta_0)	
	 		\\
		\mathcal{Q}_{z} (r, \theta, r_0, \theta_0)
			\coloneqq & \
						A(r, \theta) G_z^{reg}(r, \theta, r_0, \theta_0) B(r_0, \theta_0)
			\\
			=& \ A(r, \theta) 
			\begin{pmatrix}
			  G_{z}^{reg,+}(r, \theta, r_0,\theta_0) & 0) \\
			  0 & G_{z}^{reg, -}(r, \theta, r_0, \theta_0)
			\end{pmatrix}
			 B(r_0, \theta_0) \,,
	\end{align}
	with 
	\begin{multline}\label{eq:Green_fun_split_d}
		G_{z}^{reg,\pm} (r, \theta, r_0, \theta_0)
			\coloneqq 
			\bigg [
				\ee^{-\ii \alpha \theta} \,
				{G}_{z}(r, \theta, r_0, \theta_0) \,
				\ee^{\ii \alpha \theta_0} 
				 	\\					
				+C_{\alpha}  
					\left(4 K_{\nu^{\pm}} (\sqrt{-z} r) 
						K_{\nu^{\pm}} (\sqrt{-z} r_0)
						- \Gamma(\nu^{\pm})^2
							\left( 
								\frac{-z rr_0}{4}
							\right)^{-\nu^{\pm}}
					\right) 
					\ee^{\ii (\mp\nu^{\pm}-\alpha) (\theta - \theta_0)}
			\bigg ]  \,.
	\end{multline}
\end{subequations}
Here $\nu^+ := 1-\alpha$ and $\nu^- := \alpha$.

%
To analyse the regular part~$Q_z$, 
we need the following fact about  
the zero'th modified Bessel function~$K_0$.
\begin{Lemma} \label{le:K_0_rewriting}	
	There exist functions  $f,g$ analytic
	on $\Cb\setminus (-\infty, 0]$ and continuous at $0$ such that 		
	$$ 
	  K_0(w)= \log(w)f(w)+g(w) \,. 
	$$
	The functions $f$ and $g$ can be chosen bounded in absolute value 
	on $\Re w > 0$
	by $C_1 \ee^{-C_2 \Re w}$ with some constants $C_{1}, C_2>0$.
	Moreover, $f(0)=1$, $f(w)-f(0) = \mathcal{O}(w^2)$
	and $g(w) = \mathcal{O}(1)$ as $|w|\rightarrow 0$.
\end{Lemma}
\begin{proof}	
	Recalling the series \cite[\S\,9.6.12.--13]{AS64}
	\begin{align}\label{eq:K_0_series}
		K_0(w) = -(\log (w/2) + \gamma) I_0 
				+ \sum_{k\geq 1} (1+1/2+1/3+\cdots+1/k)\frac{(w/2)^{2k}}{(k!)^2}
		\quad
		\text{ and }	
		\quad
		I_0(w) = \sum_{k\geq 0} \frac{(w/2)^{2k}}{(k!)^2} \,,
	\end{align}	
	we define 
	$f(w)\coloneqq -I_0 \ee^{-2w}$ and $g(w) = K_0(w)-\log(w/2)f(w)$.
	With this choice $f$ is entire and $g$ is analytic outside of the 
	cut $(-\infty, 0]$.
	Note that $g$ is continuous at $0$
	and all the claimed properties follow from~\eqref{eq:K_0_series},
	 definitions of $f$ and $g$
	 and the asymptotics of $I_0, K_0$ for large arguments 
	 \cite[\S\,9.7.1--2]{AS64}.
\end{proof}

First of all, we argue that the Birman--Schwinger operator
associated with 
$H = U_\alpha (-\Delta_{a_\alpha}) U_\alpha^{-1}$
(recall~\eqref{eq:Unitary})
is regular, in agreement with its subcriticality. 

\begin{Lemma}\label{le:Q_HS1}
	Suppose Assumption~\ref{as:potential}.
	There exists a positive constant~$C$ such that,
	for all $z\in \Com\setminus [0, +\infty)$,
$$
  \| A (H-z)^{-1}  B \|_\mathrm{HS} \leq C
  \,.
$$	
\end{Lemma}
\begin{proof}
Recall that the integral kernel of $A (H-z)^{-1} B$
reads $A(x)G_z(x, x_0)B(x_0)$, where the Green function~$G_z$
is given in~\eqref{eq:plain_Green_fun}.
	We separate the proof of finiteness of the Hilbert--Schmidt norm
	\begin{align*}
	\| A (H-z)^{-1} B \|_\mathrm{HS} =
		\int_{\Rb^2 \times \Rb^2} |V(x)| \, |G_z (x, x_0)|^2 \, |V(x_0)| \,
		\dd x \dd x_0
	\end{align*}	 
	into two parts:
	\begin{enumerate}
	 \item showing that it is bounded uniformly in $|z|\in (0, 1]$,
	 \item showing that is is bounded uniformly for any $|z|>1$.
\end{enumerate}
\textbf{Ad 1.} We start with the bound for small $|z|$, 
as some of the estimates will be useful in proving the second point too.
		We add and subtract the logarithmic singularity $\log(\sqrt{-z}) f(0)$ 
		in the Green function of the Friedrichs extension $G_{z}$ given by
		\eqref{eq:plain_Green_fun} and using 
		Lemma~\ref{Lem.residual},
		we divide it in two summands
\begin{align*}
	G_{z}(x, x_0)
		&= \hat{C} F_1(|x-x_0|) +F_2(x, x_0) 		
\end{align*}
with 
\begin{align*}
	F_1(t) &\coloneqq 
			[K_0(\sqrt{-z} t)-\log(\sqrt{-z})f(0)]	\,, \quad t\geq 0 \,,  \\
	F_2(x, x_0) &\coloneqq 
		\frac{\sin(\pi \alpha)}{\pi} 
			\int_{-\infty}^{\infty} \frac{1}{2\pi} 
			F_1(R(s))
			\frac{\ee^{-\alpha s + \ii \alpha (\theta - \theta_0)}}{1 + \ee^{-s + \ii (\theta - \theta_0)}} \dd s  \,.
\end{align*}
Here~$f$ is as in Lemma~\ref{le:K_0_rewriting} and it follows that
 	\begin{align} \label{eq:F_1_bound}
 	\begin{split}
 		|F_1(t)|
 		&=
 			|\log(\sqrt{-z}) (f(\sqrt{-z} t) - f(0)) + g(\sqrt{-z} t) 
 				+\log t f(\sqrt{-z} t)|	\\					
 		&\leq 
 			C_1(|\log(\sqrt{-z})| |\sqrt{-z} t|^{\beta} 
 			+\log t +1)  \,, 
 	\end{split}			
 	\end{align}
 	for some constant $C_1>0$ and any $\beta\in (0,1)$.
 	Let us remark, that for any fixed $t\in (0, \infty)$ this stays 
 	bounded uniformly in $|z|\in (0, 1)$.

 	Turning our attention for a moment to $F_2$, we remark that  	 
 	for a fixed difference $\theta-\theta_0\neq \pm \pi$, 
 	we have the exponential decay for large $s$ of the fraction
 	\begin{align*}
 		\bigg|
 			\frac{\ee^{-\alpha s + \ii \alpha (\theta - \theta_0)}}
 	 			{1+\ee^{-s + \ii (\theta - \theta_0)}}
 	 	\bigg|^2
 	 	=
 	 	\frac{ \ee^{-2\alpha s} }
 	 		{1+2\ee^{-s}\cos(\theta-\theta_0)+ \ee^{-2s}}
 	 		\leq C \ee^{-2|s| \min\{\alpha, 1-\alpha\}} 
 	\end{align*}
 	with some $C>0$.
 	We also notice that
 	$|\frac{1}{2}\log(R(s)^2)| \leq 
 		|s| + |\log (r+r_0)|$
 	and that 
 	$R(s)^{\beta} \leq (r+r_0)^{\beta} \ee^{\beta |s|}$.
 	Hence  with an arbitrary choice 
 	$0<\beta< \min \{\alpha, 1-\alpha\} $ 
 	we conclude by \eqref{eq:F_1_bound}
 	that the integral in $F_2$
 	is convergent for every fixed pair $r, r_0$. 
	
	Due to these estimates, it is only left to show the finiteness  
		\begin{align*}
			\int_{\Rb^2\times \Rb^2} |V(x)| \,
				\big|1+\log(a)+a^{\beta} \big|^2 \, |V(x_0)| \, 
				\dd x \dd x_0 <\infty  
			\quad
			\text{for}
			\quad
			a\in\{r+r_0, |x-x_0|\} \,.	 	
		\end{align*} 
		Since  
		\begin{align}\label{eq:log_ineq}
	 		|\log a|
	 		\leq C_\beta
	 			\begin{cases}
	 				(1+a^{-\beta}) & \text{ if } a \in(0, 1] \,, \\
	 				(1+a^{\beta}) & \text{ if }  a\in [1, \infty) \,,
	 			\end{cases} 
		\end{align}
		where $C_{\beta}$ is some constant dependent only on $\beta$,
		it is enough to bound the three integrals 
		\begin{align} \label{eq:four_terms}
			\int_{\Rb^2\times \Rb^2} |V(x)| 
			\left.
				\begin{cases}
					1	\\
					a^{-2\beta}	\\
					a^{2\beta} \\
				\end{cases}
			\right\} \,
			|V(x_0)| \dd x \dd x_0 
			\quad
			\text{for}
			\quad
			a\in\{r+r_0, |x-x_0|\} \,.	 	
		\end{align}
		Finiteness of the first integral with the constant middle term
		is a direct consequence of Assumption~\ref{as:potential}.
		To treat the positive power, notice that since 
		$\beta\leq \max\{\alpha, 1-\alpha\} = \nu$ we have
		\begin{align*}
		 	|x-x_0|^{2\beta}
		 		\leq (r+r_0)^{2\beta}
		 		\leq 2^{\nu} (1+|x|^{2\nu}) (1+|x_0|^{2\nu})
		\end{align*}				
		and the integrability follows by our Assumption~\ref{as:potential}.
		
		For the negative power we first estimate
		$(r+r_0)^{-\beta} \leq |x-x_0|^{-\beta} \leq 1+ |x-x_0|^{-\mu}$ for any $\mu \geq \beta$.
		Then the integrability follows under our assumptions
		by the Hardy--Littlewood--Sobolev inequality 
		\cite[Thm.~4.3]{LL01} choosing $\beta \leq \mu = \frac{\delta}{1+\delta}$.
		Here, $\delta>0$ is as in~\eqref{eq:V_assumptions}.

	To finish the proof we point out that 
	$G_z$ is continuous at $\theta - \theta_0 = \pm \pi$, 
	see Remark~\ref{Rem.three}. 
\medskip \\
\textbf{Ad 2.} 
By the analyticity of $K_0(\sqrt{-z}t)$ on $\Re\sqrt{-z}>0$ with $t>0$, 
by the behaviour for large arguments \cite[\S\,9.7.2]{AS64} 
and by Lemma~\ref{le:K_0_rewriting}, we can bound
\begin{align*}
	|K_0(\sqrt{-z}t)| \leq C (|\sqrt{-z}t)|^{-\beta} + 1)
	\qquad \text{for all} \qquad
	t\in (0, \infty), \
	\Re\sqrt{-z}>0 \,,
\end{align*}
with some constant $C>0$ and arbitrary $\beta>0$.
For all $t\in (0, \infty)$ this is uniformly bounded in $|z|>1$.
The statement is then a consequence 
of the finiteness of the top two integrals in \eqref{eq:four_terms}, 
which was shown in the first part of the proof.
\end{proof}
\begin{Lemma}\label{le:L_Q_HS}
	Suppose Assumption~\ref{as:potential}.  
	The operators $L_z$ and $Q_z$  
 are Hilbert--Schmidt
	for all $z\in \Com \setminus [0, +\infty)$.
\end{Lemma}
\begin{proof}
Recall the formulae for~$L_z$ and~$Q_z$ given in~\eqref{eq:Green_fun_split}. 
The claim for~$L_z$ is obvious. 
	Taking Lemma~\ref{le:Q_HS1} into account, 
	the statement for~$Q_z$ follows directly from definitions 
	\eqref{eq:Green_fun_split} by the analyticity of the Bessel 
	functions on the complex half-plane 
	with positive real part and their decay 
	properties~(\cf \cite[Eqs.~(9.6.2), (9.6.10)]{AS64})
	\begin{align}\label{eq:K_nu_decay}
	K_{\mu}(w) = \mathcal{O}(\ee^{-w}/\sqrt{w})
	\qquad
	\text{as}
	\qquad
	 |w|\rightarrow \infty
\end{align}	
for $\Re{w}>0$ and $\mu\in\Rb$.
\end{proof}

%
\begin{Proposition}\label{Prop:compact_res_difference}
	Suppose Assumption~\ref{as:potential}.
	For negative~$z$ with all sufficiently large~$|z|$, 
	the difference  
	\begin{align*}
		(H_\alpha + \epsilon V -z)^{-1} -  (H_\alpha-z)^{-1}
	\end{align*}
	is a compact operator.
	Moreover, 
	$A(H_\alpha - z)^{-1/2}$ and $B(H_\alpha - z)^{-1/2}$ 
	are also compact.
\end{Proposition}
\begin{proof}
By Lemma~\ref{le:L_Q_HS} and Proposition~\ref{prop:Green_fun_Pauli},
the Birman--Schwinger operator~$R_{z, \epsilon}$ 
introduced in~\eqref{BS} and decomposed in~\eqref{BS.decomposed}
is a Hilbert--Schmidt operator for all $z\in(-\infty,0)$.
Moreover, its integral kernel $\mathcal{R}_{z, \epsilon}$ 
	tends pointwise to zero as $z\rightarrow -\infty$.
	The second part of the proof of Lemma~\ref{le:Q_HS1} 
	and~\eqref{eq:K_nu_decay}
	justify using the dominated convergence on the integral kernel
	$\mathcal{R}_{z, \epsilon}$
	to deduce 
	$\|R_{z,\epsilon}\|_\mathrm{HS} 
	= \|A (H_\alpha-z)^{-1}A\|_\mathrm{HS} 
	\rightarrow 0$ as $z\rightarrow -\infty$.
	Thus, along the lines of \cite[Ex.~7 of Sec.~XIII.4]{RS4-78}, 
	we can argue that the difference of resolvents
	\begin{align*}
		(H_\alpha + \epsilon V -z)^{-1} -  (H_\alpha-z)^{-1}
		 = -\sum _{n=0}^{\infty} \epsilon^{n+1}
		 (H_\alpha-z)^{-1} B (-A (H_\alpha-z)^{-1} B)^{n} A(H_\alpha-z)^{-1}
	\end{align*}
	is a compact operator for some negative~$z$ with large~$|z|$.
	For this conclusion we use the above observed fact that 
	$P^{\ast}P$ with
	$P\coloneqq A (H_\alpha-z)^{-1/2}$
	tends to zero in the Hilbert--Schmidt norm 
	as $z\rightarrow -\infty$.
	It then follows that 
	$A(H_\alpha-z)^{-1/2}$ and $B(H_\alpha-z)^{-1/2}$ 
	are also compact operators.
\end{proof}

%
\begin{Remark}\label{Rem.BS}
It follows from Proposition~\ref{Prop:compact_res_difference}
that $A(H_\alpha - z)^{-1/2}$ and $B(H_\alpha - z)^{-1/2}$ 
are bounded operators for all negative~$z$ with sufficiently large~$|z|$.
It justifies the usage of the Birman--Schwinger principle
in the spirit of~\cite{HK22}.
Moreover, the perturbation~$V$ is relatively form bounded 
with respect to~$H_\alpha$. By making~$\epsilon$ small,
the relative bound can be made arbitrarily small.
This justifies the sum $H_\alpha + \epsilon V$,
which should be understood in the sense of forms.
\end{Remark}

Combining Proposition~\ref{Prop:compact_res_difference}
with \cite[Thm.~XIII.14]{RS4-78}, 
we obtain the stability of the essential spectrum.
\begin{Corollary} \label{cor:stability_ess_spec}
Suppose Assumption~\ref{as:potential}. Then
	\begin{align*}
	 	\sigma_{\mathrm{ess}}(H_\alpha + \epsilon V)
	 		= \sigma_{\mathrm{ess}}(H_\alpha)
	 		= [0, +\infty) \,.
	\end{align*}
\end{Corollary}

Finally, we establish a uniform version of Lemma~\ref{le:L_Q_HS}.
\begin{Lemma} \label{Le:Q_is_HS}
	Suppose Assumption~\ref{as:potential}.
	There exists a positive constant~$C$ such that, 
	for all $z\in \Com \setminus [0,+\infty)$,
	$\|Q_z\|_\mathrm{HS} \leq C$.
	At the same time, given any positive~$\eps$, there exists a positive constant~$C_\eps$,
	such that, for all $z \in \Com \setminus [0, +\infty)$ with $|z| \geq \eps$,
	$\|L_z\|_\mathrm{HS} \leq C_\eps$.
\end{Lemma}

\begin{proof}
The claim for~$L_z$ is obvious from~\eqref{eq:L_z}. 
The uniform boundedness is also clear for the part of~$Q_z$ 
coming from the first line of~\eqref{eq:Green_fun_split_d}
due to Lemma~\ref{le:Q_HS1}.
At the same time, it follows from the structure of 
the second line of~\eqref{eq:Green_fun_split_d}
and the asymptotic behaviour~\eqref{eq:K_nu_decay} that,
given any positive~$\eps$, there exists a positive constant~$C_\eps$,
	such that, for all $z \in \Com \setminus [0, +\infty)$ with $|z| \geq \eps$,
	$\|Q_z\|_\mathrm{HS} \leq C_\eps$.
It remains to analyse the asymptotic behaviour 
of the second line of~\eqref{eq:Green_fun_split_d}
as $|z|\rightarrow 0$.

Let $\mu\in \{\alpha, 1-\alpha\}$. We establish a convenient notation
(\cf~\cite[\S\,9.6.2, \S\,9.6.10]{AS64})
\begin{align}\label{eq:K_nu_def}
	K_{\mu}(w) = w^{-\mu} f_{-\mu}(w) - w^{\mu} f_{\mu}(w)
	\qquad 	
	\text{with}
	\qquad
	f_{\mu}(w) = \frac{-\pi 2^{-\mu}}{2\sin(\mu \pi)}
	\sum_{k=0}^{\infty} \frac{(w/2)^{2k}}{k! \Gamma(k+\mu+1)} \,,
\end{align}
$w\in \Cb\setminus (-\infty, 0)$.
On account of Lemma~\ref{le:Q_HS1},
it is enough to bound
$K_{\mu}(\sqrt{-z}|x|)K_{\mu}(\sqrt{-z}|x_0|) - f^2_{-\mu}(0) |-zxx_0|^{-\mu}$
by a constant multiple of 
$(|x|^{\mu}+|x|^{-\mu})(|x_0|^{\mu}+|x_0|^{-\mu})$.
To that end we denote by $\xi=\sqrt{-z}|x|$ and by $\zeta = \sqrt{-z}|x_0|$ and
write the exact identity
\begin{align*}
	K_{\mu}(\xi)K_{\mu}(\zeta) - f^2_{-\mu}(0) (\xi\zeta)^{-\mu}
	= (K_{\mu}(\xi)-f_{-\mu}(0) \xi^{-\mu} ) K_{\mu}(\zeta)
	 +\xi^{-\mu} f_{-\mu}(0) (K_{\mu}(\zeta) - \zeta^{-\mu} f_{-\mu}(0)) \,.
\end{align*}  

If both $|\zeta|, |\xi| \geq 1$ then the left-hand side is bounded by a constant
by analyticity of $K_{\mu}$ and the decay \eqref{eq:K_nu_decay}.
If one of the arguments is small, assume
without loss of generality $|\xi| < 1$, we
notice that
$f_{\mu}(\xi)-f_{\mu}(0)= \mathcal{O}(\xi^2)$ as $|\xi|\rightarrow 0$ and use the bound
\begin{align*}	
	|K_{\mu}(\xi) - \xi^{-\mu} f_{-\mu}(0)|\leq C |\xi|^{\mu} 
\end{align*}
with some positive constant~$C$. 
If $\zeta \geq 1$ we have
$|K_{\mu}(\zeta)|\leq C_1 \leq C_1 \, |\zeta|^{\mu}$ with $C_1>0$, while for 
$|\zeta|<1$ by \eqref{eq:K_nu_decay} it holds
\begin{align*}
	|K_{\mu}(\zeta) - \zeta^{-\mu} f_{-\mu}(0)|\leq C |\zeta|^{\mu} 
	\quad \text{ and } \quad
	|K_{\mu}(\zeta)|\leq C_2 |\zeta|^{-\mu} \,,
\end{align*}
where $C, C_2>0$ are some constants.
Symmetrically we can find bounds in case $|\zeta|<1$. For any $\zeta, \xi$ with positive real part we can thus estimate
\begin{align*}
 	|K_{\mu}(\xi)K_{\mu}(\zeta) - f^2_{-\mu}(0) (\xi \zeta)^{-\mu}|
 		\leq 
 		C_3 \, (|\xi/\zeta|^{\mu} +|\zeta/\xi|^{\mu} +|\xi|^{\mu} 
 		+1)\,,
\end{align*}
for some constant $C_3>0$. In particular this stays bounded for any fixed $x, y$ as $|z|\rightarrow 0$.
Since $\mu\leq \max\{\alpha, 1-\alpha\} = \nu$ implies
$\int_{\Rb^2} |V(x)| (|x|^{2\mu}+ |x|^{-2\mu}) \leq 
2\int_{\Rb^2} |V(x)|(|x|^{2\nu}+ |x|^{-2\nu}) $
we conclude, taking Lemma~\ref{le:Q_HS1} into account, that 
under our assumptions on the potential
the operator $Q_z$ is Hilbert--Schmidt as $|z| \rightarrow 0$.
\end{proof}
%

\section{The weakly coupled eigenvalues}\label{Sec.weak}
%
In this section, we establish Theorem~\ref{Thm.main}
as a consequence of its stronger variant.

First of all, we claim that provided that there exist 
eigenvalues of the perturbed operator $H_\alpha+\epsilon V$
for all small~$\epsilon$, 
they correspond to the singularities 
in the unperturbed Green function and therefore 
necessarily tend to zero as the positive parameter~$\epsilon$ vanishes.
The fact that zero is the only possible accumulation point 
is not obvious, because we allow $H_\alpha+\epsilon V$
to be non-self-adjoint.

\begin{Lemma} \label{le:decay_of_eigenvalues}
Suppose Assumption~\ref{as:potential}. 
Let $z_\epsilon \in \sigma_\mathrm{disc}(H_\alpha+\epsilon V)$
for all sufficiently small positive~$\epsilon$. 
	Then 
	$|z_\epsilon|\rightarrow 0$ as $\epsilon \rightarrow 0$.
\end{Lemma}
\begin{proof}
By the Birman--Schwinger principle, 
there exists a normalised $\psi_\epsilon \in \sii(\Real^2,\Com^2)$
such that $R_{z,\epsilon}\psi_\epsilon=-\psi_\epsilon$ 
for all sufficiently small positive~$\epsilon$. 
Then
$$
  1 = |\langle \psi_\epsilon, R_{z, \epsilon} \psi_\epsilon \rangle|
  \leq \|R_{z, \epsilon}\|  
  \leq \|R_{z, \epsilon}\|_\mathrm{HS}  
  = \epsilon \, \|L_z + Q_z\|_\mathrm{HS}  
  \leq \epsilon \, (\|L_z\|_\mathrm{HS} + \|Q_z\|_\mathrm{HS}) 
$$
By contradiction, assume that there is a sequence 
$\{\epsilon_j\}_{j \in \Nat}$ converging to zero 
and a sequence of eigenvalues 
$\{z_{\epsilon_j}\}_{j \in \Nat}$ converging to 
a positive point~$k$ of the essential spectrum $[0,+\infty)$. 
Then the inequality above 
together with Lemma~\ref{Le:Q_is_HS} implies
$
  1 \leq \epsilon_j \, (C_{k/2} + C)
  \,,
$
where $C_{k/2}$ and $C$ are the constants from Lemma~\ref{Le:Q_is_HS},
independent of~$j$.
This is obviously a contradiction for all sufficiently large~$j$.
\end{proof}

Our next step is to reformulate the Birman--Schwinger 
principle in the usual way using the decomposition~\eqref{BS.decomposed}
of the Birman--Schwinger operator~$R_{z,\epsilon}$
into the singular part~$\epsilon L_z$ 
and the regular part~$\epsilon Q_z$.  
The existence of eigenvalue~$-1$ for~$R_{z,\epsilon}$ 
is equivalent to the lack of invertibility of
\begin{align*}
 	(1 + \epsilon (Q_z +L_z) )
 		= (1+\epsilon Q_z)(1+\epsilon (\epsilon Q_z + 1)^{-1}L_z) \,.
\end{align*}
Here the operator
 $1+\epsilon Q_z$ is invertible for all 
 sufficiently small~$\epsilon$ by Lemma~\ref{Le:Q_is_HS}.
That means that, provided that~$\epsilon$ is sufficiently small,
$-1$~is an eigenvalue of~$R_{z,\epsilon}$ if, and only if,
$-1$~is an eigenvalue of the rank-one operator
$\epsilon(\epsilon Q_z + 1)^{-1}L_z$.
 
To find the form of an eigenvalue $\lambda \neq 0$ of the
operator $\epsilon(\epsilon Q_z + 1)^{-1}L_z$,
let us denote by
$\psi$ the corresponding normalised eigenvector. Then 
by definition of $L_z$ (recall \eqref{eq:L_z})
we have (using the complex formalism $w := x^1+\ii x^2$
and $w_0 := x_0^1+\ii x_0^2$) 
\begin{align} \label{eq:eval_-1}
	\lambda \psi(w) = \epsilon (\epsilon Q_z + 1)^{-1} A(w) \overline{D(w)} \int_{\Cb} Y_z D(w_0) B(w_0) \psi(w_0) \dd w_0 \,.
\end{align}
Here we have introduced the decomposition of the integral kernel 
\begin{align*}
	\mathcal{L}_z(w, w_0) = A(w)\overline{D(w)}Y_z D(w_0)B(w_0)
\end{align*} 
using
\begin{equation}\label{eq:Y}
 	Y_z \coloneqq
 	\begin{pmatrix}
 			(-z)^{\alpha-1} & 0 \\ 
 			0 & (-z)^{-\alpha}
 	\end{pmatrix}		
	\,,  
 	\quad
 	D(w) \coloneqq
 		\sqrt{C_{\alpha}}
 		\begin{pmatrix}
 			\Gamma(1-\alpha)(|w|/2)^{\alpha-1} \ee^{\ii \, \ph (w)} & 0 \\ 
 			0 & \Gamma(\alpha)(|w|/2)^{-\alpha}
 		\end{pmatrix}			
 			\,.
\end{equation}
We rewrite \eqref{eq:eval_-1} as 
\begin{align} \label{eq:psi}
	\lambda \psi (w)= \epsilon (\epsilon Q_z + 1)^{-1} A(w) \overline{D(w)} b_z   
	&& 	\text{with}
	&&	b_z:= Y_z \int_{\Cb} D(w_0) B(w_0) \psi(w_0) \dd w_0\in \Cb^2 \,.
\end{align}
Inserting~$\psi$ back into a $\lambda$-multiple of~\eqref{eq:eval_-1}, 
we get the equation
\begin{align}\label{simplification}
 	\lambda \epsilon (\epsilon Q_z + 1)^{-1} A(w) \overline{D(w)} b_z 
 		= 
 	\epsilon (\epsilon Q_z + 1)^{-1} A(w) \overline{D(w)}
 	\epsilon Y_z \Wb(\epsilon) b_z 
\end{align}
with the matrix
$$
  \Wb (\epsilon)
 		\coloneqq \int_{\Cb} D(w_0) B(w_0) (\epsilon Q_z + 1)^{-1} A(w_0) \overline{D(w_0)} \dd w_0 \,.
$$
Applying to both sides of~\eqref{simplification}
the invertible operator $\epsilon Q_z + 1$ and dividing by~$\epsilon$,
we see that any non-zero eigenvalue~$\lambda$ of 
$\epsilon(\epsilon Q_z + 1)^{-1}L_z$ satisfies
\begin{align}\label{simplification1}
    A(w) \overline{D(w)}
 	\epsilon Y_z \Wb(\epsilon) b_z 
 	= \lambda  A(w) \overline{D(w)} b_z
\end{align}
which is a generalised eigenvalue problem in~$\Com^2$.
The following proposition summarises the above analysis
and additionally argues that~\eqref{simplification1} is equivalent 
to the usual eigenvalue problem by ``dividing by'' 
the matrix-valued function $w \mapsto A(w) \overline{D(w)}$.

\begin{Proposition}
Suppose Assumption~\ref{as:potential}. 
For all sufficiently small~$\epsilon$, 
$z \in \Com \setminus [0,+\infty)$ is an eigenvalue of
$H_\alpha + \epsilon V$  
if, and only if, $-1$~is an eigenvalue of the matrix 
$\epsilon Y_z \Wb(\epsilon)$.
\end{Proposition}
\begin{proof}
If a non-zero vector $b_z \in \Com^2$ 
solves the matrix eigenvalue problem $\epsilon Y_z \Wb(\epsilon)b_z=-b_z$,
then it is easy to check that the function~$\psi$
defined by the first formula of~\eqref{eq:psi} with $\lambda=-1$
solves  $\epsilon(\epsilon Q_z + 1)^{-1}L_z \psi = -\psi$.
Assuming $\psi=0$ implies $A(w) \overline{D(w)} b_z = 0$
for almost every $w \in \Com$.
But then $\Wb(\epsilon)b_z=0$, 
because of the structure of the matrix $\Wb(\epsilon)$,
which is impossible.

Conversely, assume $\epsilon(\epsilon Q_z + 1)^{-1}L_z \psi = -\psi$
with a non-trivial function~$\psi$.
Then~\eqref{eq:psi} holds with $\lambda=-1$
and thus defined vector~$b_z$ is necessarily non-zero. 
Applying the matrix $\epsilon Y_z \Wb(\epsilon)$ 
to~$b_z$ as defined by the integral formula of~\eqref{eq:psi}, 
it is easy to see that 
$\epsilon Y_z \Wb(\epsilon) b_z = -(\epsilon Y_z \Wb(\epsilon))^2 b_z$.
Consequently, either~$b_z$ solves  
$\epsilon Y_z \Wb(\epsilon) b_z = - b_z$
or $\epsilon Y_z \Wb(\epsilon) b_z = 0$.
Because of~\eqref{simplification1} with $\lambda=-1$,
the latter implies $A(w) \overline{D(w)} b_z = 0$
and subsequently~\eqref{eq:psi} yields $\psi=0$,
a contradiction.  
\end{proof}

By virtue of the proposition, the eigenvalue problem 
for the differential operator~$H_\alpha$ 
is reduced to analysing the matrix  eigenvalue problem 
\begin{align} \label{eq:reformulation_to_alg_problem}
	-f = \epsilon Y_z \Wb(\epsilon) f \,,
\end{align}
where $f=(f_1,f_2) \in \Com^2$.
Using the definition of~$Y_z$, this is equivalent to the coupled equations
\begin{align*}
	-f_1 &=
		\epsilon (-z)^{\alpha-1} (\Wb_{11}(\epsilon)f_1 + \Wb_{12}(\epsilon)f_2) 
		\,,
		\\
	-f_2 &=
		\epsilon (-z)^{-\alpha} (\Wb_{22}(\epsilon)f_2 + \Wb_{21}(\epsilon)f_1) 
		\,.
\end{align*}
In the case that either $\Wb_{12}(\epsilon)$ or $\Wb_{21}(\epsilon)$ 
is non-zero, this pair of equations has a solution if, and only if, 
there is a solution $z$ of the problem
\begin{align} \label{eq:main_eq_for_z}
	0= \epsilon \, [z^{-\alpha} \Wb_{22}(\epsilon) + (-z)^{\alpha -1} \Wb_{11}(\epsilon)] + \epsilon^2(-z)^{-1} \det \Wb(\epsilon) +1 \,.
\end{align}
Consider the diagonal case 
when both $\Wb_{12}(\epsilon)= \Wb_{21}(\epsilon) =0$. 
Then we obtain the decoupled system 
\begin{align*}
	-f_1 &=
		\epsilon (-z)^{\alpha-1} \Wb_{11}(\epsilon)f_1 \,,	 \\
	-f_2 &=
		\epsilon (-z)^{-\alpha} \Wb_{22}(\epsilon)f_2 \,.
\end{align*}
We conclude that there are two solutions $(f_1, 0)$ and $(0, f_2)$ 
with respective values of~$z_1$ and~$z_2$ solving this system.
It follows that in also this case we can again reformulate 
the problem of existence of a solution to \eqref{eq:reformulation_to_alg_problem} in terms of existence of a root of~\eqref{eq:main_eq_for_z}.

Following the ideas of \cite{CS18}, 
we now separate the matrix $\Wb(\epsilon)$ in two pieces 
$\Wb(\epsilon) = U + U_1(\epsilon)$ using
\begin{align} \label{eq:U}
\begin{split}
	U
	&\coloneqq 
		\int_{\Cb} D(w) B(w) A(w)\overline{D(w)} \dd w
		=\int_{\Cb} D(w) V(w)\overline{D(w)} \dd w \,,
	\\
	U_1(\epsilon)
	&\coloneqq 
		\int_{\Cb} D(w) B(w) 
		\left( [(\epsilon Q_z + 1)^{-1}-1] A\overline{D} \right)(w)
		\, \dd w 
		\,.
\end{split}
\end{align}
\begin{Lemma}\label{Lem.integrand}
Suppose Assumption~\ref{as:potential}.
Then $\|U_1(\epsilon)\| = \mathcal{O}(\epsilon)$ as $\epsilon \rightarrow 0$.
\end{Lemma}
\begin{proof}
Denoting by $\mathcal{U}_1(w)$ the integrand of $U_1(\epsilon)$,
the Cauchy--Schwarz inequality on $\Cb^2$ implies
\begin{align*}
 	\left |\int_{\Cb} \mathcal{U}_1(w) \dd w \right |
 		= 
 			\sup \left\{ 
 					\left|
 						\left \langle 
 							\psi, \int_{\Cb} \mathcal{U}_1(w) \dd w \, \phi 					\right \rangle
 					\right |
 					\, : \,
 					\phi, \psi \in \Cb^2 \,,
 					 \|\phi\|= \|\psi\| = 1
 				\right \}
 		\leq 
 			\int_{\Cb} |\mathcal{U}_1(w)| \dd w \,.
\end{align*}
Then the smallness of the norm of $U_1(\epsilon)$ follows from
the upper bound
\begin{align*}
	|{D(w)}B(w)[(1+\epsilon Q_z)^{-1} - 1] A(w) \overline{D(w)}|
	&\leq 
	|{D(w)} B(w)|\cdot \|(1+\epsilon Q_z)^{-1} - 1\|_\mathrm{HS}
	\cdot | A(w)\overline{D(w)}|\\
	&\leq \|(1+\epsilon Q_z)^{-1} - 1\|_\mathrm{HS}
	\cdot |D(w)|^2 |V(w)| 
	\,,
\end{align*}
yielding

\begin{align*}
	\|U_1(\epsilon)\| 
	&\leq 
		\frac{\epsilon \|Q_z\|_\mathrm{HS}}{1-\epsilon \|Q_z\|_\mathrm{HS} }
		\cdot C_{\alpha}
		\max\{\Gamma^2(\alpha), \Gamma^2(1-\alpha)\} 
		\int_{\Cb} (\max\{|w|^{\alpha-1}, |w|^{-\alpha}\})^2 v(w) \dd w 	\nonumber
	\\	
	&\leq	
		\frac{\epsilon \|Q_z\|_\mathrm{HS}}{1-\epsilon \|Q_z\|_\mathrm{HS} }
		C_{\alpha}
		\max\{\Gamma^2(\alpha), \Gamma^2(1-\alpha)\} 
		\int_{\Cb} (|w|^{-2\nu} + |w|^{2\nu}) v(w) \dd w 	\,.
\end{align*}
\end{proof}

For conciseness, let us write  
\begin{align*}
	a_{\epsilon} &:=(U+U_1(\epsilon))_{11}\,,
	& a_0 &:=U_{11} \,, 
	\\
	b_{\epsilon} &:=(U+U_1(\epsilon))_{22} \,,
	& b_0 &:=U_{22}\,,
    \\
	c_{\epsilon} &:=(U+U_1(\epsilon))_{11} (U+U_1(\epsilon))_{22}-(U+U_1(\epsilon))_{12}(U+U_1(\epsilon))_{21} \,,
	& c_0 &:=U_{11}U_{22}-U_{12}U_{21} \,.
\end{align*}
By Lemma~\ref{Lem.integrand},  
\begin{align}\label{eq:coefficients_aprox}
 	a_{\epsilon} = a_0 + \mathcal{O}(\epsilon)\,,
 	&&
 	b_{\epsilon} = b_0 + \mathcal{O}(\epsilon) \,,
 	&&
 	c_{\epsilon} = c_0 + \mathcal{O}(\epsilon) \,,
\end{align}
as $\epsilon \to 0$.
Then the equation~\eqref{eq:main_eq_for_z} reads
\begin{align}
\label{eq:e_quadratic}
	 \epsilon^2 c_{\epsilon} (-z)^{-1} 
	+ \epsilon ( a_{\epsilon} (-z)^{-1+\alpha} + b_{\epsilon} (-z)^{-\alpha})
	+1 = 0 \,.
\end{align}

Let us summarise our findings in the following proposition.
\begin{Proposition}\label{prop:first_equation}
Suppose Assumption~\ref{as:potential}. 
For all sufficiently small~$\epsilon$, 
$z \in \Com \setminus [0,+\infty)$ is an eigenvalue of
$H_\alpha + \epsilon V$  
if, and only if, $z$~is a root of~\eqref{eq:e_quadratic}.
\end{Proposition}

In this way, the eigenvalue problem for a differential operator
has been reduced to an implicit equation.
Since we have not been able to systematically analyse~\eqref{eq:e_quadratic}
in the general case (particular results can be derived, of course), 
let us restrict to the case of 
diagonal potentials $V = \diag(V_{11}, V_{22})$.

\begin{Theorem}\label{thm:asymptotics}
Suppose Assumption~\ref{as:potential} and assume that~$V$ is diagonal.
	\begin{enumerate}
	 	\item	
		 	If $V_{11} \neq 0$, assume $a_0 \neq 0$
		 	and $\ph\left(-\int_{\Cb} V_{11}(w)|w^2|^{\alpha -1} \dd w \right) \in (1-\alpha)(-\pi, \pi)$. 
		 	Then the operator $H_\alpha+\epsilon V$ 
		 	possesses for all sufficiently small $\epsilon> 0$ 
		 	a discrete eigenvalue~$z_+(\epsilon)$ with the asymptotics 
		\begin{align*}
		 	z_+(\epsilon) = -(-\epsilon a_{\epsilon})^{\frac{1}{1-\alpha}} 
		 		= -(-\epsilon a_0)^{\frac{1}{1-\alpha}} + \mathcal{O}(\epsilon^{\frac{2-\alpha}{1-\alpha}})
   \qquad\text{as} \qquad
		 		\epsilon \rightarrow 0 \,.
		\end{align*}
	
		\item 
			If $V_{22}\neq 0$, assume $b_0\neq 0$
			and $\ph\left(-\int_{\Cb} V_{22}(w)|w^2|^{\alpha} \dd w \right) \in \alpha(-\pi, \pi)$. Then the operator $H_\alpha+\epsilon V$ 
			possesses for all sufficiently small $\epsilon> 0$ 
		 	a discrete eigenvalue~$z_-(\epsilon)$ with the asymptotics 
			\begin{align*} 	
		 	z_-(\epsilon) = -(-\epsilon b_{\epsilon})^{\frac{1}{\alpha}}	
		 		= -(-\epsilon b_0)^{\frac{1}{\alpha}} + \mathcal{O}(\epsilon^{\frac{1+\alpha}{\alpha}})  
		 		\qquad \text{as} \qquad
		 		\epsilon \rightarrow 0 \,.
		\end{align*}
	\end{enumerate}
If both $V_{11}\neq 0$ and $V_{22}\neq 0$ satisfy the assumptions above,
then there are no other discrete eigenvalues $H_\alpha+\epsilon V$
for all sufficiently small~$\epsilon$.
If $V_{11}\neq 0$ (respectively, $V_{22}\neq 0$)
satisfies the assumptions from item~1 (respectively, item~2)
but $V_{22}=0$ (respectively, $V_{11}= 0$),
then $z_+(\epsilon)$ (respectively, $z_-(\epsilon)$)
is the unique discrete eigenvalue of $H_\alpha+\epsilon V$
for all sufficiently small~$\epsilon$.	
\end{Theorem}
\begin{proof}
If $V$ is diagonal, then so is $\mathbb{W}$, 
and $c_{\epsilon} = a_{\epsilon} b_{\epsilon}$. 
This enables us to factorise the equation \eqref{eq:e_quadratic} as
	\begin{align*}
	 	0 
	 		= (\epsilon a_{\epsilon}(-z)^{-1+\alpha} +1)(\epsilon b_{\epsilon}(-z)^{-\alpha} +1) 
	\end{align*}
	and we immediately obtain two solutions~$z_\pm$ satisfying
		\begin{align*}
	 	(-z_+)^{1-\alpha} 
	 		= -\epsilon a_{\epsilon}
	 	 \quad \text{ and } \quad
	 	(-z_-)^{\alpha} = -\epsilon b_{\epsilon} \,.
	\end{align*} 
	Under our assumptions on the potential these equations 
	have solutions $z_\pm \in \Cb\setminus [0, +\infty)$ 
	for all $\epsilon$ small enough.
	The expansions for $\epsilon \rightarrow 0$ then follow from 
	the Taylor expansions of $z_\pm$ and 
	using~\eqref{eq:coefficients_aprox}.
\end{proof}

Note that the eigenvalues~$z_+(\epsilon)$ and~$z_-(\epsilon)$ 
are eigenvalues of $H_\alpha^+ + \epsilon V_{11}$ 
and $H_\alpha^- + \epsilon V_{22}$, respectively.
If both~$V_{11}$ and~$V_{22}$ are real-valued, non-trivial and non-positive,
we obtain Theorem~\ref{Thm.main} from the introduction.

\subsection*{Acknowledgements}
Thanks belong to Johannes Ageskov and Mat\v{e}j Tu\v{s}ek 
for helpful discussions on some technical details. 
M.F.\ would further like to acknowledge support for research on this paper from 
the European Union’s Horizon 2020 research and innovation programme under the Marie Sk{\l}odowska-Curie grant
agreement No 101034413
as well as support by funding from VILLUM FONDEN through the
QMATH Centre of Excellence grant nr.~10059. 
D.K.\ was supported by the EXPRO grant number 20-17749X 
of the Czech Science Foundation (GA\v{C}R).

%
\bibliographystyle{amsplain}
\bibliography{Paulibib05} 

\end{document}